\documentclass[11pt]{article}
\usepackage[a4paper,margin=1in]{geometry}
\usepackage[T1]{fontenc}
\usepackage[utf8]{inputenc}
\usepackage{amsmath,amssymb,amsthm,mathtools}
\usepackage{graphicx}
\usepackage{anyfontsize}
\usepackage{colortbl}
\usepackage{xcolor}
\usepackage{steinmetz}
\usepackage{esvect}
\usepackage[hidelinks]{hyperref}
\graphicspath{{./images/}}
\allowdisplaybreaks

\newcommand{\trn}{^{\scriptscriptstyle T}}

\newcommand{\om}{\omega}
\newcommand{\ep}{\varepsilon}
\newcommand{\de}{\delta}
\newcommand{\rank}{\hbox{rank\,}}

\newcommand{\la}{\lambda}

\newcommand{\ga}{\gamma}

\newtheorem{theorem}{Theorem}[section]
\newtheorem{lemma}[theorem]{Lemma}
\newtheorem{statement}[theorem]{Claim}
\newtheorem{cor}[theorem]{Corollary}
\theoremstyle{remark}
\newtheorem{rem}[theorem]{Remark}

\title{On orbital stabilization of a circular motion primitive for a dynamic extension of the Dubins car model}
\author{%
Artem Angelchev-Shiryaev$^{1}$,\quad Pavel E. Aleshin$^{2}$,\quad Anton S. Shiriaev$^{3}$, \\ Pavel A. Shamanaev$^{2}$,\quad and\quad Leonid B. Freidovich$^{4}$\\[1ex]
\small $^{1}$Department of Industrial and Mechanical Sciences, Lund University, Lund, Sweden\\
\small $^{2}$Department of Information Technologies, Sirius University, Sirius, Russia\\
\small $^{3}$Department of Engineering Cybernetics, NTNU, Trondheim, Norway\\
\small $^{4}$Department of Applied Physics and Electronics, Ume\r{a} University, Ume\r{a}, Sweden\\
}
\date{}

\begin{document}
\maketitle

\begin{abstract}
This paper addresses orbital stabilization of a circular motion primitive for a dynamic extension of the Dubins car model within a transverse-linearization framework.
We show that the corresponding transverse linearization is unstable and not stabilizable by linear state feedback.
Therefore, the standard linearization-based approach to orbital stabilization cannot be applied directly.
The main contribution is a set of explicit and verifiable conditions that characterize when a controller design based on transverse linearization remains applicable.
These conditions rely on the specific structure of the dynamics in a neighborhood of the motion and on the use of non-standard transverse coordinates for controller design and analysis.
Numerical simulations illustrate the proposed design procedure.

\end{abstract}

\section{Introduction and Problem Settings}\label{sec1}

Designing a smooth feedback controller for orbital stabilization of a motion of a controlled mechanical system can often be achieved by stabilizing the origin of suitable transverse dynamics or its linear approximation; see, for example, \cite{HauserBanaszuk,NielsenMaggiore,Urabe}.
For controlled mechanical systems, constructive procedures for deriving transverse coordinates and the corresponding transverse linearization are available in a number of important cases; see \cite{Shiriaev,FreidovichShiriaevCDC09,Satre2021}.
As shown in this paper, for mechanical systems subject to non-integrable constraints, this linearization-based route may fail to provide a stabilizable transverse linearization, which leads to non-standard problem settings for motion control based solely on the linearization.
Moreover, the transverse linearization may also exhibit an unstable linear drift term, which makes the problem even more challenging.
Taken together, these two properties -- the lack of stabilizability and the instability of the linearized transverse dynamics -- raise the question of whether orbital stability of the motion can be achieved for the nonlinear mechanical system at all by means of a smooth feedback controller. 
The main result (Theorem~\ref{Thm4:main_result}), established through the analysis of an illustrative example, shows that {\em orbital stabilization of the motion can be achieved for the nonlinear system\/} even under these non-standard conditions and despite the structural properties of the transverse linearization.
In addition to proving orbital stability of the motion, the result reveals links between distinct areas of nonlinear control theory and analytical mechanics:
\begin{itemize}
\item
on the one hand, the example illustrates the classical challenge of analyzing the dynamics of a nonholonomic system in a neighborhood of the motion by means of linearization \cite{NF:1972};
\item
on the other hand, the example shows how a center manifold of the dynamics can be characterized in a neighborhood of the motion \cite{Kuznetsov}.
\end{itemize}
The paper is organized as follows.
This section continues with a description of the controlled mechanical system under study, namely, a dynamic extension of the Dubins car model, examines several properties of its dynamics in a neighborhood of one of its basic maneuvers, which forms part of the classical Dubins path \cite{Dubins,Laumond}, and formulates the problem under consideration.
The main contributions of the paper are presented in Sections~\ref{Sec:transverse_coordinates_and_linearization} and~\ref{Sec:Main_Result}.
Results of numerical simulations are reported in Section~\ref{Sec:Result_CS}.
Concluding remarks are given in Section~\ref{Sec:Concluding_Remarks}.
For readability and completeness, proofs of the nontrivial statements are collected in the Appendix.
 

\subsection{Dynamics of Dubins car model}
As an illustrative example, we consider the classical nonholonomic mechanical system commonly known as the Dubins car \cite{Bloch,Kolmanovsky,DeLuca}.
It represents a simplified model of a mobile platform in which
\begin{itemize}
    \item the position of the center of mass in the plane is described by two coordinates, $(x,y)\in\mathbb{R}^2$;
    \item the orientation of the platform is described by the heading angle $\theta\in\mathbb{S}^1$.
\end{itemize}
Then the model partly reformulates the non-sliding assumption imposed on the robot motion through the differential relations
\begin{equation}\label{kinematics_for_XY}
    \dot x(t) = v(t)\cos\bigl(\theta(t)\bigr), \qquad
    \dot y(t) = v(t)\sin\bigl(\theta(t)\bigr),
\end{equation}
where $v(t)$ denotes the longitudinal velocity of the center of mass of the robot.

Another part of the model describes the angular velocity of the mobile platform.
For planar motion, the direction of the angular velocity is fixed and orthogonal to the plane, while its magnitude is given by the rate of change of the heading angle, that is,
\begin{equation}\label{kinematics_for_Theta}
    \dot{\theta}(t) = \omega(t).
\end{equation}
The classical Dubins car model is defined by \eqref{kinematics_for_XY}--\eqref{kinematics_for_Theta} together with the assumption that the longitudinal velocity of the center of mass is constant, that is,
\begin{equation}\label{kinematics_constraint_for_V}
    v(t) = v_0, \qquad \forall t,
\end{equation}
where $v_0$ is a constant.

To derive a simple dynamic model for the coordinates
\[
q \triangleq \bigl[\theta; x; y\bigr]
\]
of the mechanical system subject to \eqref{kinematics_for_XY}--\eqref{kinematics_constraint_for_V}, we use the following facts:
\begin{enumerate}
    \item the non-sliding condition \eqref{kinematics_for_XY} implies the identity
    \begin{equation}\label{Nonholonomic_constraint}
        \dot y(t)\cos\bigl(\theta(t)\bigr) - \dot x(t)\sin\bigl(\theta(t)\bigr) = 0,
        \qquad \forall t,
    \end{equation}
    which holds for every motion of the system;

    \item if only the constraint \eqref{Nonholonomic_constraint} is taken into account and $\theta(t)$ is treated as the only directly controlled variable, then the dynamics of the three coordinates can be described by the Newton--Euler equations\footnote{Here any additional constraint, if present and deduced from \eqref{kinematics_for_XY} and \eqref{kinematics_constraint_for_V}, is temporarily ignored.} \cite{Bloch,Kolmanovsky}
    \begin{equation}\label{NE_eqns_with_R}
        \left\{
        \begin{array}{rcl}
               m\ddot x(t) &=& R_x\bigl(q(t),\dot q(t)\bigr),\\[1mm]
               m\ddot y(t) &=& R_y\bigl(q(t),\dot q(t)\bigr),\\[1mm]
               J\ddot \theta(t) &=& R_{\theta}\bigl(q(t),\dot q(t)\bigr) + u(t),
        \end{array}
        \right.
    \end{equation}
    where $m$ is the mass of the mobile platform, $J$ is its moment of inertia about the axis orthogonal to the plane, $R_x(\cdot)$, $R_y(\cdot)$, and $R_{\theta}(\cdot)$ are the components of the generalized constraint force associated with \eqref{Nonholonomic_constraint}, and $u(t)$ is the control input (torque).
\end{enumerate}

As is well known, \eqref{Nonholonomic_constraint} and \eqref{NE_eqns_with_R} are sufficient to derive the equations of motion of the constrained system in closed form.
\begin{lemma}\label{Lemma_on_constraint_force}
The dynamics of the constrained system \eqref{Nonholonomic_constraint}, \eqref{NE_eqns_with_R} is described by
\begin{equation}\label{Constrained_dynamics}
\left\{
\begin{array}{rcl}
\ddot x &=& -\bigl[\dot y\sin(\theta) + \dot x\cos(\theta)\bigr]\dot\theta\sin(\theta),\\[1mm]
\ddot y &=& \phantom{-}\bigl[\dot y\sin(\theta) + \dot x\cos(\theta)\bigr]\dot\theta\cos(\theta),\\[1mm]
\ddot\theta &=& \dfrac{1}{J}u.
\end{array}
\right.
\end{equation}
\end{lemma}

In deriving the dynamics \eqref{Constrained_dynamics} of the auxiliary system \eqref{Nonholonomic_constraint}, \eqref{NE_eqns_with_R}, we temporarily ignored both \eqref{kinematics_for_XY} and \eqref{kinematics_constraint_for_V}.
Let us now verify that every motion of the Dubins car model is a solution of \eqref{Constrained_dynamics}.
\begin{lemma}\label{Lem:V=const}
Let $\bigl[\theta(t);x(t);y(t)\bigr]$ be a forced solution of \eqref{Constrained_dynamics}.
Then the magnitude of the velocity vector $\bigl[\dot x(t);\dot y(t)\bigr]$ is constant.
\end{lemma}

\begin{lemma}\label{Lem:solutions_of_(1)_are_solutions_of_(8)}
Any set of $C^2$ functions $\bigl[\theta(t);x(t);y(t)\bigr]$ satisfying \eqref{kinematics_for_XY}--\eqref{kinematics_constraint_for_V} is a solution of the nonlinear system \eqref{Constrained_dynamics} for some control input $u(t)$.
\end{lemma}

It is worth noting that \eqref{Constrained_dynamics} is an overparameterized model for the motion of the Dubins car.
For example, if an initial state
\(\bigl[\theta(t_0);x(t_0);y(t_0);\dot\theta(t_0);\dot x(t_0);\dot y(t_0)\bigr]\)
does not satisfy the constraint \eqref{Nonholonomic_constraint}, then the corresponding solution of \eqref{Constrained_dynamics} starting from this state at $t=t_0$ does not represent a feasible motion of the Dubins car, regardless of the control input $u(t)$.


\subsection{Parametrization of a Dubins car motion primitive and problem formulation}

In this paper, we consider one of the standard maneuvers of the Dubins car: a motion in which the coordinates $x(t)$ and $y(t)$ remain on a circle of radius $r_c>0$, while the heading angle $\theta(t)$ is tangent to the curve. 
As expected, it admits a compact description.

\begin{lemma}\label{Lemma_on_motions_on_circle}
Any feasible motion of the controlled mechanical system \eqref{Constrained_dynamics} on a circle of constant radius $r_c$ centered at the origin admits the nested\footnote{In this representation, time explicitly defines the evolution of one coordinate, often called a motion generator, while the remaining coordinates vary as functions of the motion generator.} representation
\begin{equation}\label{Lem2:XYTh_c}
    \theta_c(t) = \omega_0 t + \theta_0, \qquad
    \left\{
    \begin{array}{rcl}
        x_c(t) &=& \phantom{-}r_c\sin\bigl(\theta_c(t)\bigr),\\[1mm]
        y_c(t) &=& -r_c\cos\bigl(\theta_c(t)\bigr).
    \end{array}
    \right.
\end{equation}
Here $\theta_0$ and $\omega_0$ are constants. The control input generating the motion \eqref{Lem2:XYTh_c} is zero, that is, $u_c(t) \equiv 0$.
\end{lemma}


The above parametrization of circular motion primitives and the representation of the system dynamics allow us to formulate the problem under study.
Given a forced motion $q^*(t)\triangleq \bigl[\theta^*(t); x^*(t); y^*(t)\bigr]$ of the Dubins car model whose center of mass remains on a circle of radius $r_c$ and whose heading angle remains tangent to the curve\footnote{According to Lemma~\ref{Lemma_on_motions_on_circle}, such a motion is parameterized by two scalars, $\theta_0^*$ and $\omega_0^*$, where $\theta_0^*$ determines the initial position of the Dubins car on the circle and $\omega_0^*$ defines the constant angular velocity; see \eqref{Lem2:XYTh_c}.} for all $t$,
the tasks are as follows:
\begin{description}
    \item[A:] introduce a set of functions of the system state, called transverse coordinates, that characterize deviations of a perturbed motion from the nominal one in directions transverse to $q^*(t)$;
    \item[B:] linearize the dynamics of the transverse coordinates along the nominal motion $q^*(t)$;
    \item[C:] study the properties of the transverse linearization and of the transverse coordinates defined for the Dubins car dynamics in a neighborhood of the nominal motion $q^*(t)$; and
    \item[D:] propose constructive methods for synthesizing feedback controllers aimed at stabilization of the nominal motion $q^*(t)$.
\end{description}


\section{Motion-dependent transverse coordinates and their variational dynamics}\label{Sec:transverse_coordinates_and_linearization}

Given constants $\theta_0^*$ and $\om_0^*$ defining the nominal motion $q^*(t)$ of the Dubins car along the circle of radius $r_c$, see \eqref{Lem2:XYTh_c}, consider the following five scalar functions of the state vector
$X \triangleq \bigl[\theta; x; y; \dot\theta; \dot x; \dot y\bigr]:$
\begin{equation}\label{Transverse_coordinates}
\begin{array}{rcl@{\qquad}rcl}
x_{1_\bot}(X) &\triangleq& x - r_c\sin(\theta),
&
x_{2_\bot}(X) &\triangleq& y + r_c\cos(\theta), \\[1mm]
x_{3_\bot}(X) &\triangleq& \dot x - r_c\cos(\theta)\dot\theta,
&
x_{4_\bot}(X) &\triangleq& \dot y - r_c\sin(\theta)\dot\theta,
\qquad
x_{5_\bot}(X) \triangleq \dot\theta - \om_0^*.
\end{array}
\end{equation}
Taking into account the parametrization of the nominal motion \eqref{Lem2:XYTh_c}, one readily verifies that each of the functions defined in \eqref{Transverse_coordinates} vanishes on $q^*(t)$, that is when $X=X^*\triangleq \bigl[\theta^*; x^*; y^*; \dot\theta^*; \dot x^*; \dot y^*\bigr]$.
Hence, they are natural candidates for transverse coordinates associated with $q^*(t)$.
As shown below, these five functions are functionally independent and therefore provide one possible choice for representing transverse directions to the nominal motion.
For future reference, we refer to them as {\em the standard set of transverse coordinates} for the motion \eqref{Lem2:XYTh_c}.

\begin{lemma}\label{Lem3:Jacobian}
Given the nominal motion $q^*(t)$ of the controlled mechanical system \eqref{Constrained_dynamics} along the circle of radius $r_c$, parameterized by $\theta_0=\theta_0^*$ and $\om_0=\om_0^*$, see \eqref{Lem2:XYTh_c}, let
$X_\bot \triangleq \bigl[x_{1_\bot}; \dots; x_{5_\bot}\bigr]$
be the vector function with components defined in \eqref{Transverse_coordinates}.
Then the Jacobian $J_\bot(X)$ of the mapping from the state
$X=\bigl[\theta; x; y; \dot\theta; \dot x; \dot y\bigr]$
to $X_\bot(X)$ is the $5\times 6$ matrix
\begin{equation}\label{Lem3:J}
J_\bot(X)=
\left[
\begin{array}{cccccc}
 \phantom{-}j_1(X)              & 1 & 0 & 0      & 0 & 0 \\
 \phantom{-}j_2(X)              & 0 & 1 & 0      & 0 & 0 \\
 -j_2(X)\dot\theta              & 0 & 0 & j_1(X) & 1 & 0 \\
 \phantom{-}j_1(X)\dot\theta    & 0 & 0 & j_2(X) & 0 & 1 \\
 0                              & 0 & 0 & 1      & 0 & 0
\end{array}
\right],
\end{equation}
where, for brevity,\quad
$j_1(X) \triangleq -r_c\cos(\theta),$\quad
$j_2(X) \triangleq -r_c\sin(\theta).$\quad
Moreover, the rank of $J_\bot(X)$ is constant and equal to $5$ for every state vector $X$.
\end{lemma}

Any non-singular transformation of the vector function $X_\bot(X)$ on the nominal motion defines a new set of five transverse coordinates for $q^*(t)$.

Since each of the five functions $x_{i_\bot}(X)$, $i=1,\dots,5$, defined in \eqref{Transverse_coordinates}, vanishes on the nominal motion $q^*(t)$, and since their variations in a neighborhood of the motion are linearly independent\footnote{As shown in Lemma~\ref{Lem3:Jacobian}, the Jacobian \eqref{Lem3:J} has full rank.}, these five functions form a full set of transverse coordinates for the motion.
As a natural next step, we compute the linearization of the transverse coordinates \eqref{Transverse_coordinates} along the nominal motion.
The following statement gives the result.

\begin{theorem}\label{Thm1:TL}
The linearization of the transverse coordinates \eqref{Transverse_coordinates} along the nominal motion
$q^*(t)=\bigl[\theta^*(t);x^*(t);y^*(t)\bigr],$
parameterized by the constants $\theta_0^*$ and $\om_0^*$, see \eqref{Lem2:XYTh_c}, of the nonlinear control system \eqref{Constrained_dynamics}, has the form
\begin{equation}\label{Transverse_linearization}
    \hbox{$\frac{d}{dt}$}{\de X}_\bot = A(t)\de X_\bot + B(t)\de u,
\end{equation}
\[\hbox{where}\qquad
\de X_\bot \triangleq
\left[\!\!
\begin{array}{c}
\de x_{1_\bot}\\
\de x_{2_\bot}\\
\de x_{3_\bot}\\
\de x_{4_\bot}\\
\de x_{5_\bot}
\end{array}
\!\!\right],
\qquad
A(t) =
\left[\!\!
\begin{array}{ccccc}
0 & 0 & 1 & 0 & 0\\
0 & 0 & 0 & 1 & 0\\
0 & 0 & a_{33}(t) & a_{34}(t) & 0\\
0 & 0 & a_{43}(t) & a_{44}(t) & 0\\
0 & 0 & 0 & 0 & 0
\end{array}
\!\!\right],
\qquad
B(t) =
\left[\!\!
\begin{array}{c}
0\\
0\\
b_3(t)\\
b_4(t)\\
b_5(t)
\end{array}
\!\!\right].
\]
The nontrivial entries of $A(t)$ and $B(t)$ are $\left(2\pi/\om_0^*\right)$-periodic functions given by
\begin{equation*}
\begin{array}{rcl@{\qquad}rcl}
a_{33}(t) &=& -\om_0^*\sin\bigl(\om_0^* t+\theta_0^*\bigr)\cos\bigl(\om_0^* t+\theta_0^*\bigr),
&
a_{34}(t) &=& -\om_0^*\sin^2\bigl(\om_0^* t+\theta_0^*\bigr), \\[1mm]
a_{43}(t) &=& \phantom{-}\om_0^*\cos^2\bigl(\om_0^* t+\theta_0^*\bigr),
&
a_{44}(t) &=& \phantom{-}\om_0^*\sin\bigl(\om_0^* t+\theta_0^*\bigr)\cos\bigl(\om_0^* t+\theta_0^*\bigr), \\[1mm]
b_3(t) &=& -\frac{r_c}{J}\cos\bigl(\om_0^* t+\theta_0^*\bigr),
&
b_4(t) &=& -\frac{r_c}{J}\sin\bigl(\om_0^* t+\theta_0^*\bigr),
\qquad
b_5(t) = \frac{1}{J}.
\end{array}
\end{equation*}
Here $\de x_{i_\bot}(t)$, $i=1,\dots,5$, and $\de u(t)$ denote the first-order variations of the corresponding quantities along the nominal motion.
\end{theorem}

The explicit form of the transverse coordinates \eqref{Transverse_coordinates} and of their variational dynamics \eqref{Transverse_linearization} in a neighborhood of the nominal motion suggests several possibilities for analysis and feedback design; see \cite{Shiriaev} and related works.
We next summarize the key properties of the linear control system \eqref{Transverse_linearization}.
\begin{lemma}\label{Lem4:invariant_subspace_of_TL}
Given constants $\theta_0^*$ and $\om_0^*$ defining the nominal motion \eqref{Lem2:XYTh_c} of the Dubins car along the circle of radius $r_c$, the variational dynamics \eqref{Transverse_linearization} of the transverse coordinates \eqref{Transverse_coordinates} has the following properties:
\begin{enumerate}
    \item For any solution
        $\de X_\bot(t)=\bigl[\de x_{1_\bot}(t);\dots;\de x_{5_\bot}(t)\bigr]$ 
    of the linear control system \eqref{Transverse_linearization}, the scalar function
    \begin{equation}\label{Lem4:Invariant}
        I\bigl(t,\de X_\bot(t)\bigr) \triangleq \de x_{4_\bot}(t)\cos\bigl(\theta^*(t)\bigr) - \de x_{3_\bot}(t)\sin\bigl(\theta^*(t)\bigr)
    \end{equation}
    remains constant for every control input $\de u(t)$.

    \item The periodic linear control system \eqref{Transverse_linearization} is not stabilizable, that is, there is no matrix function $K(t)$ such that the origin of the closed-loop system
    \begin{equation*}
        \hbox{$\frac{d}{dt}$}{\de X}_\bot(t)=A(t)\de X_\bot(t)+B(t)\de u(t), \qquad
        \de u(t)=K(t)\de X_\bot(t),
    \end{equation*}
    is asymptotically stable.

    \item The origin of the homogeneous linear system
    \begin{equation}\label{Lem4:linear_system}
        \hbox{$\frac{d}{dt}$}{\de X}_\bot(t)=A(t)\de X_\bot(t)
    \end{equation}
    is unstable, and the system has at least a one-dimensional subspace of unbounded solutions.
\end{enumerate}
\end{lemma}

%


\section{Main Result}\label{Sec:Main_Result}

The lack of stabilizability of the variational dynamics \eqref{Transverse_linearization} of the transverse coordinates \eqref{Transverse_coordinates} corresponding to the nominal motion \eqref{Lem2:XYTh_c} of the nonlinear control system \eqref{Constrained_dynamics}, established in Lemma~\ref{Lem4:invariant_subspace_of_TL}, creates an obstacle to asymptotic orbital stabilization of the nominal motion if one follows the approach of \cite{Shiriaev} or related works.
Indeed, that approach relies on the possibility of stabilizing the origin of the variational dynamics, and this property is absent in the present case.
Moreover, as stated in Lemma~\ref{Lem4:invariant_subspace_of_TL}, the uncontrolled linear system \eqref{Transverse_linearization}, with $\de u \equiv 0$, is unstable.

Nevertheless, the following result, which is the main contribution of the paper, shows that smooth state-feedback laws ensuring orbital stability of the nominal motion do exist for the considered dynamic extension of the Dubins car model, despite these non-standard problem settings.

\begin{theorem}\label{Thm4:main_result}
Let $q^*(t)=\bigl[\theta^*(t);x^*(t);y^*(t)\bigr]$ be a motion of the Dubins car \eqref{kinematics_for_XY}--\eqref{kinematics_constraint_for_V} along a circle of radius $r_c$, parameterized by the constants $\theta_0^*$ and $\om_0^*$, see \eqref{Lem2:XYTh_c}. Suppose that the nonlinear control system \eqref{Constrained_dynamics}, augmented with a $C^1$-smooth state-feedback law
\begin{equation}\label{Thm4:nonlinear_feedback}
    u = U(X), \qquad X=\bigl[\theta;x;y;\dot\theta;\dot x;\dot y\bigr],
\end{equation}
has the following property:
\begin{enumerate}
    \item[(PL)] The variational dynamics of the transverse coordinates \eqref{Transverse_coordinates}, along the nominal motion $q^*(t)$, of the closed-loop system \eqref{Constrained_dynamics}, \eqref{Thm4:nonlinear_feedback} has an invariant linear subspace of dimension three, and the restriction of the variational dynamics to this subspace is quadratically stable.
\end{enumerate}
Then the following statements hold:
\begin{enumerate}
    \item The nominal motion $q^*(t)$ of the Dubins car is orbitally stable. More precisely, for every $\varepsilon>0$, there exists $\de=\de(\varepsilon)>0$ such that, for any initial condition $X_0=\bigl[\theta_0;x_0;y_0;\dot\theta_0;\dot x_0;\dot y_0\bigr]$ of a perturbed motion $X(t)=X(t,X_0)$ starting sufficiently close to the initial point of the nominal motion, the implication
    \begin{equation}\label{Thm4:orbital_stability}
        \|X_0-X_0^*\|<\de
        \quad\Rightarrow\quad
        \hbox{\rm dist}\,\bigl\{X(t),\Gamma^*\bigr\}<\ep,
        \qquad \forall t\ge 0
    \end{equation}
    holds. Here $X_0^*=\bigl[\theta_0^*;x_0^*;y_0^*;\dot\theta_0^*;\dot x_0^*;\dot y_0^*\bigr]$ is the initial condition of the nominal motion, and $X^*(t)=X^*(t,X_0^*)$ is its state-space representation. The set $\Gamma^*$ is the curve in the state space of the closed-loop system \eqref{Constrained_dynamics}, \eqref{Thm4:nonlinear_feedback}, which is the orbit defined by the nominal motion. That is, $X=\bigl[\theta;x;y;\dot\theta;\dot x;\dot y\bigr]\in\Gamma^*$ if and only if there exists $\tau\ge 0$ such that $X=X^*(\tau)$. Finally, $\hbox{\rm dist}\,\bigl\{X,\Gamma^*\bigr\}$ denotes the distance between the point $X$ and the set $\Gamma^*$;

    \item If, in addition, the initial condition $X_0$ of the perturbed motion satisfies
    \begin{equation}\label{Thm4:v0=Om*rc}
        \sqrt{\dot x_0^2+\dot y_0^2}=\om_0^* r_c,
    \end{equation}
    then, in addition to \eqref{Thm4:orbital_stability}, the limit relation
    \begin{equation}\label{Thm4:a_orbital_stability}
        \hbox{\rm dist}\,\bigl\{X(t),\Gamma^*\bigr\}\to 0
        \qquad \hbox{as} \qquad t\to+\infty
    \end{equation}
    holds, and the convergence in \eqref{Thm4:a_orbital_stability} is exponential.
\end{enumerate}
\end{theorem}

\begin{rem}\label{rem01}
The theorem introduces a new sufficient condition, namely {\em (PL)}, formulated in terms of the variational dynamics \eqref{Transverse_linearization}. This condition allows one to infer orbital stability of the nominal motion $q^*(t)$ for the nonlinear closed-loop system.
Indeed, smoothness of the state-feedback law \eqref{Thm4:nonlinear_feedback} ensures that, locally in a neighborhood of the nominal motion, it can be written in the form
\begin{equation*}
    u = U(X) = K(X)X_\bot(X),
\end{equation*}
where $K(X)$ is a continuous gain and $X_\bot(X)$ is the vector of transverse coordinates defined in \eqref{Transverse_coordinates}; see \cite{Hartman}.
Therefore, the linearization of the transverse dynamics of the nonlinear system \eqref{Constrained_dynamics}, \eqref{Thm4:nonlinear_feedback} along the nominal motion $X^*(t)$ has the form
\begin{equation*}
    \hbox{$\frac{d}{dt}$}{\de X}_\bot(t) = A(t)\de X_\bot(t) + B(t)\de u(t), \qquad
    \de u(t) = \left.K(X)\right|_{X=X^*(t)}\de X_\bot(t).
\end{equation*}
According to property {\em (PL)}, this periodic linear system has an invariant asymptotically stable subspace of dimension three.
\end{rem}

To construct feedback laws satisfying the new sufficient condition {\em (PL)} in Theorem~\ref{Thm4:main_result}, we present the following facts about the linear control system \eqref{Transverse_linearization} expressed in different coordinates.

\begin{lemma}\label{Lem6:Linear_system_in_new_coordinates}
Consider the change of coordinates
\begin{equation}\label{Lem6:change_of_coordinates}
    \de X_\bot = L(t)\de Z,
    \qquad
    L(t)=\Phi(t)D^{-1},
    \qquad
    D=
    \left[\!\begin{array}{ccccc}
        d_{11} & 0      & 0      & 0 & 0\\
        0      & d_{22} & d_{23} & 0 & 0\\
        0      & 0      & d_{33} & 0 & 0\\
        0      & 0      & 0      & 1 & 0\\
        0      & 0      & d_{53} & 0 & 1
    \end{array}\!\right].
\end{equation}
\begin{equation*}
\hbox{Here}
    \qquad
    \de Z \triangleq
    \left[\!\begin{array}{c}
        \de z_1\\
        \de z_2\\
        \de z_3\\
        \de z_4\\
        \de z_5
    \end{array}\!\right]=
    \left[\!\begin{array}{c}
        \de z\\
        \de z_4\\
        \de z_5
    \end{array}\!\right],
    \qquad
    \de z \triangleq
    \left[\!\begin{array}{c}
        \de z_1\\
        \de z_2\\
        \de z_3
    \end{array}\!\right],
\end{equation*}
and $D$ is the constant matrix with coefficients
\begin{equation*}
    d_{11}=\frac{J\om_0^*}{r_c}, \qquad
    d_{22}=-\frac{J\om_0^*}{r_c}, \qquad
    d_{23}=d_{33}=-\frac{J}{r_c}, \qquad
    d_{53}=\frac{1}{r_c}.
\end{equation*}
The matrix function $\Phi(t)$ is defined by
\begin{equation}\label{Lem5:Phi}
    \Phi(t)=
    \left[\begin{array}{ccccc}
        1 & 0 & \frac{1}{\om_0^*}\sin(\om_0^* t+\theta_0^*) & \varphi_{14}(t) & 0\\
        0 & 1 & \frac{1}{\om_0^*}\bigl[1-\cos(\om_0^* t+\theta_0^*)\bigr] & \varphi_{24}(t) & 0\\
        0 & 0 & \cos(\om_0^* t+\theta_0^*) & \varphi_{34}(t) & 0\\
        0 & 0 & \sin(\om_0^* t+\theta_0^*) & \varphi_{44}(t) & 0\\
        0 & 0 & 0 & \varphi_{54}(t) & 1
    \end{array}\right],
\end{equation}
where the components of $\vv{\varphi}_4(t)=\bigl[\varphi_{14}(t);\dots;\varphi_{54}(t)\bigr]$ are
\begin{eqnarray}\label{Lem5:z4}
    \vv{\varphi}_4(t)
    &=&
    \left[\!\!
    \begin{array}{c}
        \sin(\om_0^* t+\theta_0^*)\\[2.5mm]
        -\cos(\om_0^* t+\theta_0^*)\\[2.5mm]
        \om_0^* \cos(\om_0^* t+\theta_0^*)\\[2.5mm]
        \om_0^* \sin(\om_0^* t+\theta_0^*)\\[2.5mm]
        0
    \end{array}\!\!
    \right] t
    +
    \left[\!\!
    \begin{array}{c}
        \frac{\theta_0^*}{\om_0^*}\sin(\om_0^* t+\theta_0^*)
        + \frac{2}{\om_0^*}\bigl[\cos(\om_0^* t+\theta_0^*)-1\bigr]\\[2mm]
        -\frac{\theta_0^*}{\om_0^*}\cos(\om_0^* t+\theta_0^*)
        + \frac{2}{\om_0^*}\sin(\om_0^* t+\theta_0^*)\\[2mm]
        \theta_0^*\cos(\om_0^* t+\theta_0^*)-\sin(\om_0^* t+\theta_0^*)\\[2mm]
        \theta_0^*\sin(\om_0^* t+\theta_0^*)+\cos(\om_0^* t+\theta_0^*)\\[2mm]
        0
    \end{array}\!\!
    \right].
\end{eqnarray}
This transformation converts the  linear periodic control system \eqref{Transverse_linearization} into
\begin{equation}\label{Transverse_linearization_in_Z}
    \hbox{$\frac{d}{dt}$}{\de Z} = A_z(t)\de Z + B_z(t)\de u,
\end{equation}
where
\begin{equation*}
    A_z(t)=
    \left[\!\begin{array}{ccccc}
        0 & 0 & 0 & 0 & 0\\
        0 & 0 & 0 & 0 & 0\\
        0 & 0 & 0 & 0 & 0\\
        0 & 0 & 0 & 0 & 0\\
        0 & 0 & 0 & 0 & 0
    \end{array}\!\right],
    \qquad
    B_z(t)=
    \left[\!\begin{array}{c}
        \sin(\om_0^* t+\theta_0^*)\\
        \cos(\om_0^* t+\theta_0^*)\\
        1\\
        0\\
        0
    \end{array}\!\right].
\end{equation*}
Thus, \eqref{Transverse_linearization_in_Z} is a linear periodic control system with zero drift.
\end{lemma}

Since the linear system \eqref{Transverse_linearization_in_Z} consists of two decoupled subsystems,
\begin{eqnarray}
    \hbox{{\bf LS1:}\qquad }\hbox{$\frac{d}{dt}$}\de z
    &=&
    \left[\!\begin{array}{c}
        \sin(\om_0^* t+\theta_0^*)\\
        \cos(\om_0^* t+\theta_0^*)\\
        1
    \end{array}\!\right]\de u, \label{Transverse_linearization_Z_2_subsystems}\\[2mm]
    \hbox{{\bf LS2:}\qquad }\hbox{$\frac{d}{dt}$}
    \left[\!\begin{array}{c}
        \de z_4\\
        \de z_5
    \end{array}\!\right]
    &=&
    \left[\!\begin{array}{c}
        0\\
        0
    \end{array}\!\right].\nonumber
\end{eqnarray}
and since the second subsystem, {\bf LS2}, is obviously not stabilizable, Lemma~\ref{Lem6:Linear_system_in_new_coordinates} clarifies the construction of a feedback law satisfying condition {\em (PL)} in Theorem~\ref{Thm4:main_result}. In particular, condition {\em (PL)} holds if and only if the control input $\de u$ stabilizes subsystem {\bf LS1} in \eqref{Transverse_linearization_Z_2_subsystems}.

Before discussing approaches to stabilizing the variational dynamics, we comment on the complexity of the task and show that no ``simple'' feedback law solves the problem.

\begin{lemma}\label{Lem8:no_static_feedback_for_TL}
Given constants $\theta_0^*$ and $\om_0^*$, consider the linear control system \eqref{Transverse_linearization}, rewritten in the form \eqref{Transverse_linearization_in_Z}, or equivalently \eqref{Transverse_linearization_Z_2_subsystems}. Then there is no linear feedback law
\begin{equation}\label{Lem8:static_feedback}
    \de u = k_1 \de z_1 + k_2 \de z_2 + k_3 \de z_3
\end{equation}
with constant gains $k_1$, $k_2$, and $k_3$ that makes subsystem {\bf LS1} of \eqref{Transverse_linearization_Z_2_subsystems} asymptotically stable.
\end{lemma}

This observation motivates the search for non-static feedback laws that stabilize a maximally stabilizable subspace of the linear control system \eqref{Transverse_linearization_Z_2_subsystems}. Such a controller can be constructed readily.
Indeed, the controllability Gramian of subsystem {\bf LS1} in \eqref{Transverse_linearization_Z_2_subsystems} is
\begin{equation}\label{Lem7:Gc}
    \Gamma_{\hbox{\small (\ref{Transverse_linearization_Z_2_subsystems})}}
    =
    \int\limits_{t_0}^{T_0+t_0}
    \left[\!\!\begin{array}{c}
        \sin(\om_0^* t+\theta_0^*)\\
        \cos(\om_0^* t+\theta_0^*)\\
        1
    \end{array}\!\!\right]
    \left[\!\!\begin{array}{c}
        \sin(\om_0^* t+\theta_0^*)\\
        \cos(\om_0^* t+\theta_0^*)\\
        1
    \end{array}\!\!\right]^{\trn}
    dt
    =
    \left[\!\!
    \begin{array}{ccc}
        \frac{\pi}{\om_0^*} & 0 & 0\\
        0 & \frac{\pi}{\om_0^*} & 0\\
        0 & 0 & \frac{2\pi}{\om_0^*}
    \end{array}\!\!
    \right],
\end{equation}
which has full rank. Hence, subsystem {\bf LS1} of \eqref{Transverse_linearization_Z_2_subsystems} is completely controllable over the period $T_0=\frac{2\pi}{\om_0^*}$. Therefore, one may follow a systematic design route and use feedback laws obtained from an auxiliary {\em linear-quadratic} optimization problem.
The advantages and computational complexity of this approach are discussed and illustrated on benchmark examples in \cite{Gusev_IJC:2016,Gusev_BIT:2010}.
However, for many examples, including the present case study, there are clear alternatives.
To make the digital implementation as transparent as possible, we illustrate the design of a state-feedback controller for specific values of the constants $\theta_0^*$ and $\om_0^*$ parameterizing the motion \eqref{Lem2:XYTh_c}.

\begin{lemma}\label{Lem9:fb_for_om=2}
Let $\theta_0^*=0~[\mathrm{rad}]$ and $\om_0^*=2~\bigl[\frac{\mathrm{rad}}{\mathrm{s}}\bigr]$. Consider subsystem {\bf LS1} of \eqref{Transverse_linearization_Z_2_subsystems} augmented with the linear state-feedback law
\begin{eqnarray}\label{Lem9:fb1}
    \de u
    &=&
    -4\sin\bigl(\om_0^* t\bigr)\de z_1
    -5\cos\bigl(\om_0^* t\bigr)\de z_2
    -3\de z_3
    \nonumber\\
    &=&
    C(t)\de Z,
    \qquad
    C(t)\triangleq
    \bigl[-4\sin\bigl(\om_0^* t\bigr),\,-5\cos\bigl(\om_0^* t\bigr),\,-3,\,0,\,0\bigr].
\end{eqnarray}
Then the corresponding closed-loop system \eqref{Transverse_linearization_Z_2_subsystems}, \eqref{Lem9:fb1} is quadratically stable.
Consequently, the linear control system \eqref{Transverse_linearization}, augmented with the linear state-feedback law
\begin{equation}\label{Lem9:fb}
    \de u = K(t)\de X_\bot,
    \qquad
    K(t)=C(t)D\bigl[\Phi(t)\bigr]^{-1},
\end{equation}
where the matrix $D$ and the matrix function $\Phi(t)$ are defined in \eqref{Lem6:change_of_coordinates}--\eqref{Lem5:Phi}, possesses a quadratically stable invariant linear subspace of maximal dimension three.
\end{lemma}

The linear state-feedback design \eqref{Lem9:fb1}--\eqref{Lem9:fb} given in Lemma~\ref{Lem9:fb_for_om=2} can be transformed into a family of nonlinear controllers satisfying condition {\em (PL)} in Theorem~\ref{Thm4:main_result}. Therefore, any such nonlinear feedback law orbitally stabilizes the nominal motion $q^*(t)$.
\begin{cor}\label{Corollary_01}\em
Let $\theta_0^*=0~[\mathrm{rad}]$ and $\om_0^*=2~\bigl[\frac{\mathrm{rad}}{\mathrm{s}}\bigr]$, and let \eqref{Lem2:XYTh_c} define the nominal motion of the Dubins car. Consider a smooth nonlinear feedback controller
\begin{equation}\label{Lem9:Nonlinear_fb}
    u = K(X)X_\bot(X),
\end{equation}
where the gain $K(X)$ satisfies the interpolation condition
\begin{equation}\label{Lem9:Nonlinear_fb1}
    \left.K(X)\right|_{X=X^*(t)} = C(t)D\bigl[\Phi(t)\bigr]^{-1},
\end{equation}
and where $X=\bigl[\theta;x;y;\dot\theta;\dot x;\dot y\bigr]$ is the state vector of \eqref{Constrained_dynamics}, $X_\bot(X)$ is the vector function with components defined in \eqref{Transverse_coordinates}, and $C(t)$, $\Phi(t)$, and $D$ are defined in \eqref{Lem9:fb1}--\eqref{Lem9:fb}. Then the closed-loop system \eqref{Constrained_dynamics}, \eqref{Lem9:Nonlinear_fb} satisfies all sufficient conditions of Theorem~\ref{Thm4:main_result}.
One such nonlinear controller \eqref{Lem9:Nonlinear_fb}--\eqref{Lem9:Nonlinear_fb1} is obtained by choosing
\begin{equation}\label{Lem9:Nonlinear_fb2}
    K(X) \triangleq C(s)D\bigl[\Phi(s)\bigr]^{-1},
    \qquad
    s=\frac{\theta-\theta_0^*}{\om_0^*},
\end{equation}
which automatically satisfies the interpolation condition \eqref{Lem9:Nonlinear_fb1}.
\end{cor} 


\section{Results of Computer Simulation}\label{Sec:Result_CS}

To illustrate the main contribution, we study the response of the nonlinear system \eqref{Constrained_dynamics} under the nonlinear feedback controller proposed in Corollary~\ref{Corollary_01} and defined by \eqref{Lem9:Nonlinear_fb} and \eqref{Lem9:Nonlinear_fb2}.
The parameters of the robot and of the nominal motion are
$m=1~[\mathrm{kg}]$, $J=1~[\mathrm{kg\,m^2}]$, $r_c=1~[\mathrm{m}]$, $\theta_0^*=0~[\mathrm{rad}]$, and $\om_0^*=2~\bigl[\frac{\mathrm{rad}}{\mathrm{s}}\bigr]$.
By Corollary~\ref{Corollary_01}, the linearization of the closed-loop system along the nominal motion $q^*(t)$ satisfies property {\em (PL)}.
Therefore, Theorem~\ref{Thm4:main_result} implies that the nonlinear controller defined by \eqref{Lem9:Nonlinear_fb} and \eqref{Lem9:Nonlinear_fb2} orbitally stabilizes the nominal motion $q^*(t)$.
The simulation results are shown in Figs.~\ref{fig1} and~\ref{fig2}. Figure~\ref{fig1} shows the evolution of the five transverse coordinates $X_\bot(t)$ along a solution of the nonlinear closed-loop system, while Fig.~\ref{fig2} shows the corresponding control signal.
For illustration, the initial condition of the perturbed motion of the Dubins car is chosen to satisfy \eqref{Thm4:v0=Om*rc}.
As stated in the second part of Theorem~\ref{Thm4:main_result}, all five transverse variables converge to zero.

\begin{figure}
    \centering
    \includegraphics[width=1.0\linewidth]{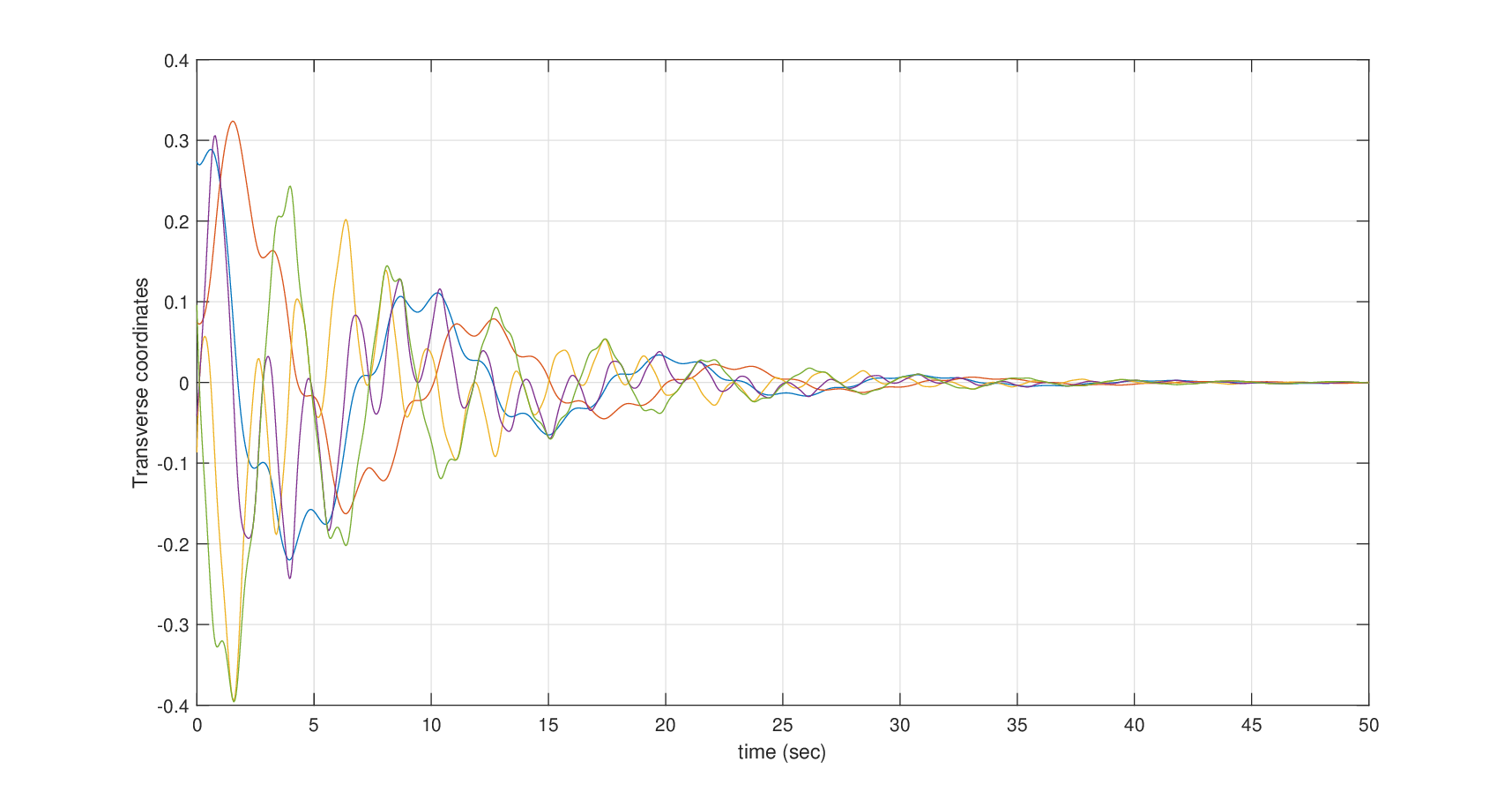}\vskip-2mm
    \caption{Time evolution of the transverse coordinates $X_\bot(t)$ along a solution of the closed-loop system. The initial condition of the perturbed motion of the Dubins car is chosen to satisfy \eqref{Thm4:v0=Om*rc}. As a result, all five transverse variables converge to zero.}
    \label{fig1}
\end{figure}

\begin{figure}
    \centering
    \includegraphics[width=0.9\linewidth]{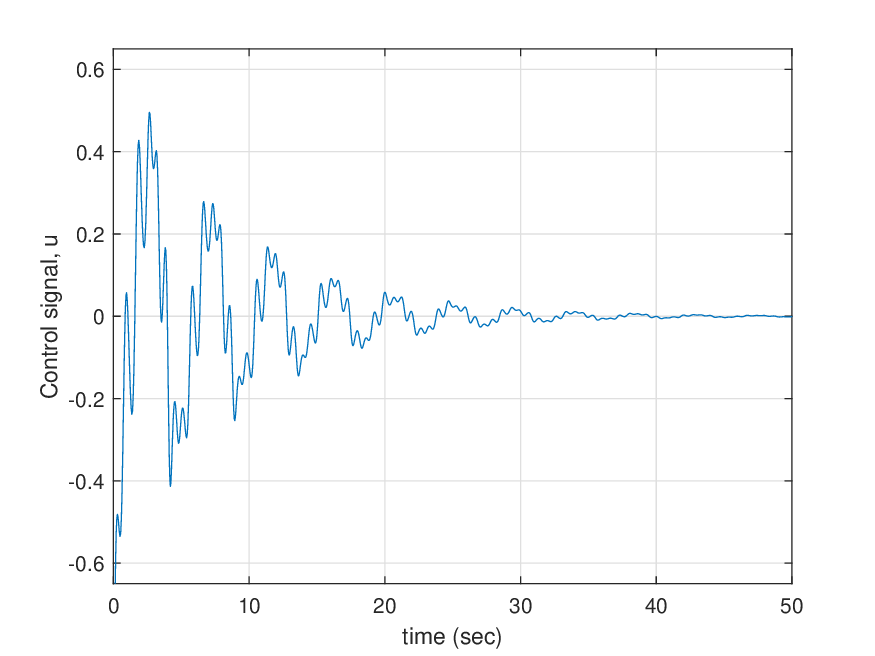}\vskip-2mm
    \caption{Control signal along a solution of the closed-loop system.}
    \label{fig2}
\end{figure}


\section{Concluding Remarks}\label{Sec:Concluding_Remarks}

This paper examined orbital stabilization of a circular motion primitive for a dynamic extension of the Dubins car model through analysis of variational dynamics of transverse coordinates. The main difficulty is that the corresponding transverse linearization is not stabilizable by linear state feedback and, in addition, its uncontrolled dynamics is unstable. Therefore, the standard transverse-linearization-based route to orbital stabilization cannot be applied directly.

To address this difficulty, we derived a new sufficient condition, denoted by {\em (PL)}, formulated in terms of the variational dynamics of the transverse coordinates. This condition requires the existence of an invariant three-dimensional subspace on which the variational dynamics is quadratically stable. For the considered Dubins car example, we showed how this condition can be verified constructively. In particular, we obtained an explicit change of coordinates that decomposes the transverse linearization and separates the non-stabilizable part of the dynamics from the part that can be stabilized by feedback.

Based on this decomposition, we constructed a linear state-feedback law for the variational dynamics and then converted it into a smooth nonlinear state-feedback law for the original constrained system. This construction made it possible to prove orbital stability of the nominal motion for the resulting nonlinear closed-loop system. Moreover, under the additional condition that the initial translational speed satisfies \eqref{Thm4:v0=Om*rc}, the distance from the perturbed motion to the nominal orbit converges to zero exponentially.

The results show that failure of stabilizability of the full transverse linearization does not, by itself, exclude successful orbital stabilization of the periodic solution of the nonlinear system. What is essential is the existence of a sufficiently large invariant stable subspace of the variational dynamics together with a nonlinear feedback law whose linearization reproduces the corresponding partially stabilizing design.

The example considered in the paper is specific, but the presented reasoning is constructive and can be reused in other problems of orbital stabilization for constrained and underactuated mechanical systems. Possible directions for future work include extending the analysis to broader classes of nonholonomic motions, clarifying how the sufficient condition {\em (PL)} can be checked systematically for more general systems, and comparing the proposed construction with alternative transverse-linearization-based and periodic-control designs.

\bibliographystyle{unsrt}
\bibliography{references_for_paper}


\appendix

\section{Proof of Lemma~\ref{Lemma_on_constraint_force}} 
\begin{proof}
The statement follows from the standard assumption that the constraint \eqref{Nonholonomic_constraint} is ideal; see, for example, \cite{Bloch}. Therefore, the instantaneous power of the corresponding reaction force vanishes along any motion:
\begin{equation*}
    R_\theta\bigl(q(t),\dot q(t)\bigr)\dot\theta(t)
    + R_x\bigl(q(t),\dot q(t)\bigr)\dot x(t)
    + R_y\bigl(q(t),\dot q(t)\bigr)\dot y(t)
    \equiv 0.
\end{equation*}
Since
$\dot q(t)=\bigl[\dot\theta(t);\dot x(t);\dot y(t)\bigr],$
the Pfaffian form corresponding to \eqref{Nonholonomic_constraint} is
\(
   \bigl[\,0,\,-\sin\theta,\ \cos\theta\,\bigr]\,\dot q = 0.
\)
Hence, the generalized reaction force is proportional to the corresponding constraint covector, and its components can be written as
\begin{equation}\label{Lem01:R=la}
    R_\theta\bigl(q(t),\dot q(t)\bigr)=0, \qquad
    R_x\bigl(q(t),\dot q(t)\bigr) = -\lambda(t)\sin\bigl(\theta(t)\bigr), \qquad
    R_y\bigl(q(t),\dot q(t)\bigr) = \lambda(t)\cos\bigl(\theta(t)\bigr),
\end{equation}
for some scalar function $\lambda(t)$.

Next, differentiate the constraint \eqref{Nonholonomic_constraint} along a motion:
\begin{align*}
    0
    &= \frac{d}{dt}\Bigl[\dot y\cos\theta - \dot x\sin\theta\Bigr] 
    = \ddot y\cos\theta - \dot y\sin\theta\,\dot\theta
       - \ddot x\sin\theta - \dot x\cos\theta\,\dot\theta.
\end{align*}
Substituting the expressions for $\ddot x$ and $\ddot y$ from \eqref{NE_eqns_with_R}, together with \eqref{Lem01:R=la}, we obtain
\begin{align*}
    0
    &= \left[\frac{\lambda}{m}\cos\theta\right]\cos\theta
       - \dot y\sin\theta\,\dot\theta
       - \left[-\frac{\lambda}{m}\sin\theta\right]\sin\theta
       - \dot x\cos\theta\,\dot\theta \\    &
       = \frac{\lambda}{m}\bigl[\cos^2\theta+\sin^2\theta\bigr]
       - \dot y\sin\theta\,\dot\theta
       - \dot x\cos\theta\,\dot\theta 
       = \frac{\lambda}{m}
       - \bigl[\dot y\sin\theta+\dot x\cos\theta\bigr]\dot\theta.
\end{align*}
Therefore,
\begin{equation*}
    \lambda(t)
    =
    m\bigl[\dot y(t)\sin\bigl(\theta(t)\bigr)+\dot x(t)\cos\bigl(\theta(t)\bigr)\bigr]\dot\theta(t).
\end{equation*}
Substituting this expression into \eqref{Lem01:R=la} and then into \eqref{NE_eqns_with_R} yields \eqref{Constrained_dynamics}.
\end{proof}
\section{Proof of Lemma~\ref{Lem:V=const}} 
\begin{proof}
The statement follows from the identity
\begin{align*}
    \frac{d}{dt}\bigl[\dot x(t)^2+\dot y(t)^2\bigr]
    &=
    2\dot x(t)\ddot x(t)+2\dot y(t)\ddot y(t) \\
    &=
    2\dot x(t)\bigl[-\la(t)\sin\bigl(\theta(t)\bigr)\bigr]
    +2\dot y(t)\bigl[\la(t)\cos\bigl(\theta(t)\bigr)\bigr] \\    &
    =
    2\la(t)\Bigl[-\dot x(t)\sin\bigl(\theta(t)\bigr)+\dot y(t)\cos\bigl(\theta(t)\bigr)\Bigr].
\end{align*}
By the constraint \eqref{Nonholonomic_constraint}, the expression in brackets is identically zero. Hence,
 \(   \frac{d}{dt}\bigl[\dot x(t)^2+\dot y(t)^2\bigr]=0.\)
Therefore, the scalar function
\[
    v(t)=\sqrt{\dot x(t)^2+\dot y(t)^2}
\]
remains constant along every solution of \eqref{Constrained_dynamics}, independently of the control input $u(t)$.
\end{proof}


\section{Proof of Lemma~\ref{Lem:solutions_of_(1)_are_solutions_of_(8)}} 
\begin{proof}
Let $\bigl[\theta^\star(t);x^\star(t);y^\star(t)\bigr]$ be a set of $C^2$ functions satisfying \eqref{kinematics_for_XY}--\eqref{kinematics_constraint_for_V}. To prove the statement, it is enough to check that $x^\star(t)$ and $y^\star(t)$ satisfy the first two equations of \eqref{Constrained_dynamics}, and then define a control input $u(t)$ so that the third equation is also satisfied.

Differentiating \eqref{kinematics_for_XY} for $x^\star(t)$ and using \eqref{kinematics_constraint_for_V}, we obtain
\begin{align*}
    \ddot x^\star(t)
    &= \frac{d}{dt}\Bigl[v\cos\bigl(\theta^\star(t)\bigr)\Bigr] 
    = \dot v \cos\bigl(\theta^\star(t)\bigr)
       - v\sin\bigl(\theta^\star(t)\bigr)\dot\theta^\star(t) 
       = -v\sin\bigl(\theta^\star(t)\bigr)\dot\theta^\star(t).
\end{align*}
Since \eqref{kinematics_for_XY} implies
\[
    v = \dot x^\star(t)\cos\bigl(\theta^\star(t)\bigr)
      + \dot y^\star(t)\sin\bigl(\theta^\star(t)\bigr),
\]
it follows that
\begin{equation*}
    \ddot x^\star(t)
    =
    -\Bigl[\dot y^\star(t)\sin\bigl(\theta^\star(t)\bigr)
      + \dot x^\star(t)\cos\bigl(\theta^\star(t)\bigr)\Bigr]
    \sin\bigl(\theta^\star(t)\bigr)\dot\theta^\star(t),
\end{equation*}
which coincides with the first equation of \eqref{Constrained_dynamics}.

Similarly,
\begin{align*}
    \ddot y^\star(t)
    &= \frac{d}{dt}\Bigl[v\sin\bigl(\theta^\star(t)\bigr)\Bigr] 
    = \dot v \sin\bigl(\theta^\star(t)\bigr)
       + v\cos\bigl(\theta^\star(t)\bigr)\dot\theta^\star(t) 
       = v\cos\bigl(\theta^\star(t)\bigr)\dot\theta^\star(t),
\end{align*}
and therefore
\begin{equation*}
    \ddot y^\star(t)
    =
    \Bigl[\dot y^\star(t)\sin\bigl(\theta^\star(t)\bigr)
      + \dot x^\star(t)\cos\bigl(\theta^\star(t)\bigr)\Bigr]
    \cos\bigl(\theta^\star(t)\bigr)\dot\theta^\star(t),
\end{equation*}
which coincides with the second equation of \eqref{Constrained_dynamics}.

Finally, since $\theta^\star(t)$ is a $C^2$ function, the third equation of \eqref{Constrained_dynamics} is satisfied by choosing
\begin{equation*}
    u(t) \triangleq J\ddot\theta^\star(t).
\end{equation*}
Hence, the triple $\bigl[\theta^\star(t);x^\star(t);y^\star(t)\bigr]$ is a solution of \eqref{Constrained_dynamics} for the control input $u(t)$ defined above.
\end{proof}

\section{Proof of Lemma~\ref{Lemma_on_motions_on_circle}} 
\begin{proof}
If a forced solution of the system \eqref{Nonholonomic_constraint}, \eqref{NE_eqns_with_R} remains on the circle of radius $r_c$ for all $t$, then it must have the form
\begin{equation}\label{Lem2:XYTh_c2}
    x_c(t)=r_c\cos\bigl(\psi(t)\bigr), \qquad
    y_c(t)=r_c\sin\bigl(\psi(t)\bigr), \qquad
    \theta_c(t)=\psi(t)+\frac{\pi}{2},
\end{equation}
where $\psi(t)$ is the polar angle of the point $\bigl(x_c(t),y_c(t)\bigr)$ on the circle.

By Lemma~\ref{Lemma_on_constraint_force}, the functions in \eqref{Lem2:XYTh_c2} must satisfy \eqref{Constrained_dynamics} for some control input $u_c(t)$. Differentiating \eqref{Lem2:XYTh_c2}, we obtain
\begin{align*}
    \dot x_c &= -r_c\sin(\psi)\dot\psi, \qquad
    \dot y_c = r_c\cos(\psi)\dot\psi, \qquad
    \dot\theta_c = \dot\psi, \\
    \ddot x_c &= -r_c\cos(\psi)\dot\psi^2-r_c\sin(\psi)\ddot\psi, \qquad
    \ddot y_c = -r_c\sin(\psi)\dot\psi^2+r_c\cos(\psi)\ddot\psi, \qquad
    \ddot\theta_c = \ddot\psi.
\end{align*}

Next, observe that along any motion of \eqref{Constrained_dynamics}, the weighted sum
\[
    \ddot x_c(t)\cos\bigl(\theta_c(t)\bigr)
    +\ddot y_c(t)\sin\bigl(\theta_c(t)\bigr)
\]
is identically zero, because the first two equations of \eqref{Constrained_dynamics} imply
$\ddot x\cos\theta+\ddot y\sin\theta = 0.$
Substituting the expressions above together with $\theta_c=\psi+\frac{\pi}{2}$, we obtain
\begin{align*}
    0
    &= \ddot x_c\cos\theta_c+\ddot y_c\sin\theta_c \\
    &= \bigl[-r_c\cos(\psi)\dot\psi^2-r_c\sin(\psi)\ddot\psi\bigr]
       \cos\Bigl(\psi+\frac{\pi}{2}\Bigr) 
       + \bigl[-r_c\sin(\psi)\dot\psi^2+r_c\cos(\psi)\ddot\psi\bigr]
       \sin\Bigl(\psi+\frac{\pi}{2}\Bigr) \\
    &= r_c\Bigl[\cos(\psi)\sin\Bigl(\psi+\frac{\pi}{2}\Bigr)
       -\sin(\psi)\cos\Bigl(\psi+\frac{\pi}{2}\Bigr)\Bigr]\ddot\psi \\    &\qquad
       -r_c\Bigl[\cos(\psi)\cos\Bigl(\psi+\frac{\pi}{2}\Bigr)
       +\sin(\psi)\sin\Bigl(\psi+\frac{\pi}{2}\Bigr)\Bigr]\dot\psi^2 \\
    &= r_c\sin\Bigl(\psi+\frac{\pi}{2}-\psi\Bigr)\ddot\psi
       -r_c\cos\Bigl(\psi+\frac{\pi}{2}-\psi\Bigr)\dot\psi^2 
       = r_c\ddot\psi.
\end{align*}
Hence,
\(
    \ddot\psi=0,
\)
and therefore
\(
    \psi(t)=c_1 t+c_0
\)
for some constants $c_0$ and $c_1$.
Define
\(
    \omega_0 \triangleq c_1, ~~
    \theta_0 \triangleq c_0+\frac{\pi}{2}.
\)
Then \eqref{Lem2:XYTh_c2} becomes
\[
    \theta_c(t)=\omega_0 t+\theta_0, \qquad
    x_c(t)=r_c\sin\bigl(\theta_c(t)\bigr), \qquad
    y_c(t)=-r_c\cos\bigl(\theta_c(t)\bigr),
\]
which is exactly the parametrization \eqref{Lem2:XYTh_c}.

Finally, since $\ddot\theta_c=\ddot\psi=0$, the third equation of \eqref{Constrained_dynamics} gives
\(
    u_c(t)=J\ddot\theta_c(t)\equiv 0.
\)
This proves both parts of the lemma.
\end{proof}


\section{Proof of Lemma~\ref{Lem3:Jacobian}} 
\begin{proof}
The explicit form of the Jacobian \eqref{Lem3:J} follows by differentiating each component of the vector function $X_\bot(X)$ defined in \eqref{Transverse_coordinates} with respect to the state vector
$X=\bigl[\theta;x;y;\dot\theta;\dot x;\dot y\bigr].$
In particular, for each $i=1,\dots,5$,
\begin{equation}\label{Lem3:dex_i}
    \delta x_{i_\bot}
    =
    \frac{\partial x_{i_\bot}}{\partial \theta}\delta\theta
    + \frac{\partial x_{i_\bot}}{\partial x}\delta x
    + \frac{\partial x_{i_\bot}}{\partial y}\delta y
    + \frac{\partial x_{i_\bot}}{\partial \dot\theta}\delta\dot\theta
    + \frac{\partial x_{i_\bot}}{\partial \dot x}\delta\dot x
    + \frac{\partial x_{i_\bot}}{\partial \dot y}\delta\dot y.
\end{equation}
Applying \eqref{Lem3:dex_i} to the five functions in \eqref{Transverse_coordinates} gives
\begin{align*}
    \delta x_{1_\bot}
    &=
    -r_c\cos(\theta)\,\delta\theta + \delta x, \\
    \delta x_{2_\bot}
    &=
    -r_c\sin(\theta)\,\delta\theta + \delta y, \\
    \delta x_{3_\bot}
    &=
    r_c\sin(\theta)\dot\theta\,\delta\theta
    -r_c\cos(\theta)\,\delta\dot\theta
    + \delta\dot x, \\
    \delta x_{4_\bot}
    &=
    -r_c\cos(\theta)\dot\theta\,\delta\theta
    -r_c\sin(\theta)\,\delta\dot\theta
    + \delta\dot y, \\
    \delta x_{5_\bot}
    &=
    \delta\dot\theta.
\end{align*}
These relations yield the Jacobian matrix \eqref{Lem3:J} with
\[
    j_1(X)\triangleq -r_c\cos(\theta),
    \qquad
    j_2(X)\triangleq -r_c\sin(\theta).
\]

To verify its rank, consider the $5\times 5$ minor formed by the last five columns of $J_\bot(X)$:
\[
    \left[
    \begin{array}{ccccc}
        1 & 0 & 0      & 0 & 0 \\
        0 & 1 & 0      & 0 & 0 \\
        0 & 0 & j_1(X) & 1 & 0 \\
        0 & 0 & j_2(X) & 0 & 1 \\
        0 & 0 & 1      & 0 & 0
    \end{array}
    \right].
\]
Its determinant is equal to $1$, independently of $X$. Therefore, $\rank J_\bot(X)=5$ for every state vector $X$.
\end{proof}

\section{Proof of Theorem~\ref{Thm1:TL}} 
\begin{proof}
The proof consists of five computational steps. In each step, we derive the variational dynamics of one of the transverse coordinates \eqref{Transverse_coordinates} in a neighborhood of the nominal motion $q^*(t)$ defined by \eqref{Lem2:XYTh_c}, for the nominal input $u^*(t)\equiv 0$.

For convenience, let
\begin{equation}\label{Thm1:xbot=}
    \de x_{i_\bot}(t)\triangleq
    \left.\frac{\partial x_{i_\bot}}{\partial \theta}\right|_{*}\de\theta
    + \left.\frac{\partial x_{i_\bot}}{\partial x}\right|_{*}\de x
    + \left.\frac{\partial x_{i_\bot}}{\partial y}\right|_{*}\de y
    + \left.\frac{\partial x_{i_\bot}}{\partial \dot\theta}\right|_{*}\de\dot\theta
    + \left.\frac{\partial x_{i_\bot}}{\partial \dot x}\right|_{*}\de\dot x
    + \left.\frac{\partial x_{i_\bot}}{\partial \dot y}\right|_{*}\de\dot y,
    \qquad i=1,\dots,5,
\end{equation}
where $\left.[\,\cdot\,]\right|_{*}$ denotes evaluation along the nominal motion and the nominal input.

\medskip
\noindent
{\bf Step 1.} Consider the first transverse coordinate
\(
    x_{1_\bot}(X)=x-r_c\sin(\theta).
\)
Along a solution $\bigl[X(t),u(t)\bigr]$ of \eqref{Constrained_dynamics}, its time derivative is
\begin{equation}\label{Thm1:dx1}
    \hbox{$\frac{d}{dt}$}\bigl[x_{1_\bot}(X(t))\bigr]
    =
    \hbox{$\frac{d}{dt}$}\bigl[x(t)-r_c\sin(\theta(t))\bigr]
    =
    \dot x(t)-r_c\cos(\theta(t))\dot\theta(t)
    \triangleq f_1(X(t),u(t)).
\end{equation}
Since $x_{1_\bot}$ vanishes on the nominal motion, we have
\begin{equation}\label{Thm1:x1}
    x_{1_\bot}(X(t))
    =
    \de x_{1_\bot}(t)+\bigl[\hbox{higher-order terms}\bigr].
\end{equation}
Next, linearize the function
\[
    f_1(X,u)=\dot x-r_c\cos(\theta)\dot\theta
\]
at the nominal motion. Since $f_1$ does not depend explicitly on $u$, the coefficient of $\de u$ is zero. Moreover,
\[
    \left.\frac{\partial f_1}{\partial x}\right|_{*}=0, \qquad
    \left.\frac{\partial f_1}{\partial y}\right|_{*}=0, \qquad
    \left.\frac{\partial f_1}{\partial \dot y}\right|_{*}=0,
\]
while
\[
    \left.\frac{\partial f_1}{\partial \theta}\right|_{*}
    =
    r_c\sin\bigl(\theta^*(t)\bigr)\dot\theta^*(t), \qquad
    \left.\frac{\partial f_1}{\partial \dot x}\right|_{*}=1, \qquad
    \left.\frac{\partial f_1}{\partial \dot\theta}\right|_{*}
    =
    -r_c\cos\bigl(\theta^*(t)\bigr).
\]
Hence,
\begin{align}
    f_1(X(t),u(t))
    &= \de\dot x
    + r_c\sin\bigl(\theta^*(t)\bigr)\dot\theta^*(t)\de\theta
    - r_c\cos\bigl(\theta^*(t)\bigr)\de\dot\theta
    + \bigl[\hbox{higher-order terms}\bigr] \nonumber\\
    &= \de x_{3_\bot}(t)+\bigl[\hbox{higher-order terms}\bigr],
    \label{Thm1:f1}
\end{align}
where the last identity follows from \eqref{Transverse_coordinates} and \eqref{Thm1:xbot=}. Substituting \eqref{Thm1:x1} and \eqref{Thm1:f1} into \eqref{Thm1:dx1}, we obtain
\begin{equation}\label{Thm1:dx1_03}
    \hbox{$\frac{d}{dt}$}\Bigl[\de x_{1_\bot}(t)+\bigl[\hbox{higher-order terms}\bigr]\Bigr]
    =
    \de x_{3_\bot}(t)+\bigl[\hbox{higher-order terms}\bigr].
\end{equation}
Therefore,
\begin{equation}\label{Thm1:Dx1}
    \hbox{$\frac{d}{dt}$}\de x_{1_\bot}(t)=\de x_{3_\bot}(t),
\end{equation}
which yields the first row of \eqref{Transverse_linearization}. In particular,
\[
    a_{11}(t)=a_{12}(t)=a_{14}(t)=a_{15}(t)=b_1(t)\equiv 0, \qquad a_{13}(t)\equiv 1.
\]

\medskip
\noindent
{\bf Step 2.} Consider the second transverse coordinate
\(
    x_{2_\bot}(X)=y+r_c\cos(\theta).
\)
Along a solution $\bigl[X(t),u(t)\bigr]$ of \eqref{Constrained_dynamics}, its time derivative is
\begin{equation}\label{Thm1:dx2}
    \hbox{$\frac{d}{dt}$}\bigl[x_{2_\bot}(X(t))\bigr]
    =
    \hbox{$\frac{d}{dt}$}\bigl[y(t)+r_c\cos(\theta(t))\bigr]
    =
    \dot y(t)-r_c\sin(\theta(t))\dot\theta(t)
    \triangleq f_2(X(t),u(t)).
\end{equation}
Since $x_{2_\bot}$ vanishes on the nominal motion, we have
\begin{equation}\label{Thm1:x2}
    x_{2_\bot}(X(t))
    =
    \de x_{2_\bot}(t)+\bigl[\hbox{higher-order terms}\bigr].
\end{equation}
Next, linearize the function
\[
    f_2(X,u)=\dot y-r_c\sin(\theta)\dot\theta
\]
at the nominal motion. Since $f_2$ does not depend explicitly on $u$, the coefficient of $\de u$ is zero. Moreover,
\[
    \left.\frac{\partial f_2}{\partial x}\right|_{*}=0, \qquad
    \left.\frac{\partial f_2}{\partial y}\right|_{*}=0, \qquad
    \left.\frac{\partial f_2}{\partial \dot x}\right|_{*}=0,
\]
while
\[
    \left.\frac{\partial f_2}{\partial \theta}\right|_{*}
    =
    -r_c\cos\bigl(\theta^*(t)\bigr)\dot\theta^*(t), \qquad
    \left.\frac{\partial f_2}{\partial \dot y}\right|_{*}=1, \qquad
    \left.\frac{\partial f_2}{\partial \dot\theta}\right|_{*}
    =
    -r_c\sin\bigl(\theta^*(t)\bigr).
\]
Hence,
\begin{align}
    f_2(X(t),u(t))
    &= \de\dot y
    - r_c\cos\bigl(\theta^*(t)\bigr)\dot\theta^*(t)\de\theta
    - r_c\sin\bigl(\theta^*(t)\bigr)\de\dot\theta
    + \bigl[\hbox{higher-order terms}\bigr] \nonumber\\
    &= \de x_{4_\bot}(t)+\bigl[\hbox{higher-order terms}\bigr],
    \label{Thm1:f2}
\end{align}
where the last identity again follows from \eqref{Transverse_coordinates} and \eqref{Thm1:xbot=}. Substituting \eqref{Thm1:x2} and \eqref{Thm1:f2} into \eqref{Thm1:dx2}, we obtain
\begin{equation}\label{Thm1:dx2_03}
    \hbox{$\frac{d}{dt}$}\Bigl[\de x_{2_\bot}(t)+\bigl[\hbox{higher-order terms}\bigr]\Bigr]
    =
    \de x_{4_\bot}(t)+\bigl[\hbox{higher-order terms}\bigr].
\end{equation}
Therefore,
\begin{equation}\label{Thm1:Dx2}
    \hbox{$\frac{d}{dt}$}\de x_{2_\bot}(t)=\de x_{4_\bot}(t),
\end{equation}
which yields the second row of \eqref{Transverse_linearization}. In particular,
\[
    a_{21}(t)=a_{22}(t)=a_{23}(t)=a_{25}(t)=b_2(t)\equiv 0, \qquad a_{24}(t)\equiv 1.
\]

\medskip
\noindent
{\bf Step 3.} Consider the third transverse coordinate
\(
    x_{3_\bot}(X)=\dot x-r_c\cos(\theta)\dot\theta.
\)
Along a perturbed solution $\bigl[X(t),u(t)\bigr]$ of \eqref{Constrained_dynamics}, its time derivative is
\begin{align}
    \hbox{$\frac{d}{dt}$}\bigl[x_{3_\bot}(X(t))\bigr]  &= \hbox{$\frac{d}{dt}$}\bigl[\dot x(t)-r_c\cos(\theta(t))\dot\theta(t)\bigr] 
    = \ddot x(t)+r_c\sin(\theta(t))\dot\theta(t)^2-r_c\cos(\theta(t))\ddot\theta(t) \nonumber\\
    &= -\Bigl[\dot y(t)\sin\bigl(\theta(t)\bigr)+\dot x(t)\cos\bigl(\theta(t)\bigr)\Bigr]
       \dot\theta(t)\sin\bigl(\theta(t)\bigr) \nonumber\\    &\qquad
       +r_c\sin\bigl(\theta(t)\bigr)\dot\theta(t)^2
       -\frac{r_c}{J}\cos\bigl(\theta(t)\bigr)u(t)
       \triangleq f_3(X(t),u(t)).
    \label{Thm1:dx3}
\end{align}
Since $x_{3_\bot}$ vanishes on the nominal motion, we have
\begin{equation}\label{Thm1:x3}
    x_{3_\bot}(X(t))
    =
    \de x_{3_\bot}(t)+\bigl[\hbox{higher-order terms}\bigr].
\end{equation}
Next, linearize the function
\[
    f_3(X,u)
    =
    -\Bigl[\dot y\sin(\theta)+\dot x\cos(\theta)\Bigr]\sin(\theta)\dot\theta
    +r_c\sin(\theta)\dot\theta^2
    -\frac{r_c}{J}\cos(\theta)u.
\]
Thus,
\begin{equation}
\begin{split}
    f_3(X(t),u(t))\label{Thm1:f3}
    &=
    a_{31}(t)\de x_{1_\bot}
    +a_{32}(t)\de x_{2_\bot}
    +a_{33}(t)\de x_{3_\bot}
    +a_{34}(t)\de x_{4_\bot}
    +a_{35}(t)\de x_{5_\bot}
    +b_3(t)\de u\\&
    +\bigl[\hbox{higher-order terms}\bigr].
		\end{split}
\end{equation}
At the same time, the first-order Taylor expansion of $f_3$ with respect to the independent state variations has the form
\begin{equation}\label{Thm1:f3_Taylor_series}
\begin{split}
    f_3(X(t),u(t))
    &=
    \left.\frac{\partial f_3}{\partial \theta}\right|_{*}\de\theta
    +\left.\frac{\partial f_3}{\partial \dot x}\right|_{*}\de\dot x
    +\left.\frac{\partial f_3}{\partial \dot y}\right|_{*}\de\dot y
    +\left.\frac{\partial f_3}{\partial \dot\theta}\right|_{*}\de\dot\theta
    +\left.\frac{\partial f_3}{\partial u}\right|_{*}\de u 
    +\bigl[\hbox{higher-order terms}\bigr],
\end{split}
\end{equation}
where
\begin{equation}\label{Thm1:paf3}
    \begin{array}{rcl}
        \displaystyle \left.\frac{\partial f_3}{\partial \dot x}\right|_{*}
        &=&
        -\cos\bigl(\theta^*(t)\bigr)\sin\bigl(\theta^*(t)\bigr)\dot\theta^*(t),\\[3mm]
        \displaystyle \left.\frac{\partial f_3}{\partial \dot y}\right|_{*}
        &=&
        -\sin^2\bigl(\theta^*(t)\bigr)\dot\theta^*(t),\\[3mm]
        \displaystyle \left.\frac{\partial f_3}{\partial \dot\theta}\right|_{*}
        &=&
        \Bigl[2r_c\dot\theta^*(t)
        -\sin\bigl(\theta^*(t)\bigr)\dot y^*(t)
        -\cos\bigl(\theta^*(t)\bigr)\dot x^*(t)\Bigr]
        \sin\bigl(\theta^*(t)\bigr),\\[3mm]
        \displaystyle \left.\frac{\partial f_3}{\partial u}\right|_{*}
        &=&
        -\frac{r_c}{J}\cos\bigl(\theta^*(t)\bigr).
    \end{array}
\end{equation}
The explicit expression for $\left.\frac{\partial f_3}{\partial \theta}\right|_{*}$ is not needed below. Equating the first-order terms in \eqref{Thm1:f3} and \eqref{Thm1:f3_Taylor_series}, we obtain
\begin{equation}\label{Thm1_identity}
\begin{split}
    &a_{31}(t)\de x_{1_\bot}
    +a_{32}(t)\de x_{2_\bot}
    +a_{33}(t)\de x_{3_\bot}
    +a_{34}(t)\de x_{4_\bot}
    +a_{35}(t)\de x_{5_\bot}
    +b_3(t)\de u \\
    &\qquad\equiv
    \left.\frac{\partial f_3}{\partial \theta}\right|_{*}\de\theta
    +\left.\frac{\partial f_3}{\partial \dot x}\right|_{*}\de\dot x
    +\left.\frac{\partial f_3}{\partial \dot y}\right|_{*}\de\dot y
    +\left.\frac{\partial f_3}{\partial \dot\theta}\right|_{*}\de\dot\theta
    +\left.\frac{\partial f_3}{\partial u}\right|_{*}\de u.
\end{split}
\end{equation}
Here the variations $\de x_\bot=\bigl[\de x_{1_\bot};\dots;\de x_{5_\bot}\bigr]$ on the left-hand side are related to the independent state variations by
\[
    \de x_\bot = J_\bot(X(t))
    \left[\!\begin{array}{c}
        \de q\\
        \de\dot q
    \end{array}\!\right],
\]
with the Jacobian matrix $J_\bot(X)$ defined in \eqref{Lem3:J}. Since the independent variations $\de x$ and $\de y$ do not appear on the right-hand side of \eqref{Thm1_identity}, while they enter the left-hand side through $\de x_{1_\bot}$ and $\de x_{2_\bot}$, it follows that
\begin{equation}\label{Thm1:a31a32=0}
    a_{31}(t)\equiv 0,
    \qquad
    a_{32}(t)\equiv 0.
\end{equation}
Similarly,
\begin{equation}\label{Thm1:b3=b3}
    b_3(t)
    =
    \left.\frac{\partial f_3}{\partial u}\right|_{*}
    =
    -\frac{r_c}{J}\cos\bigl(\theta^*(t)\bigr).
\end{equation}
Substituting \eqref{Thm1:a31a32=0} and \eqref{Thm1:b3=b3} into \eqref{Thm1_identity}, and using
\[
    \de x_{3_\bot}
    =
    \de\dot x
    +r_c\sin\bigl(\theta^*(t)\bigr)\dot\theta^*(t)\de\theta
    -r_c\cos\bigl(\theta^*(t)\bigr)\de\dot\theta,
\]\[
    \de x_{4_\bot}
    =
    \de\dot y
    -r_c\cos\bigl(\theta^*(t)\bigr)\dot\theta^*(t)\de\theta
    -r_c\sin\bigl(\theta^*(t)\bigr)\de\dot\theta,
\]
and $\de x_{5_\bot}=\de\dot\theta$, we obtain
\begin{equation}\label{Thm1_identity_02}
\begin{split}
    &a_{33}(t)\de x_{3_\bot}
    +a_{34}(t)\de x_{4_\bot}
    +a_{35}(t)\de x_{5_\bot} 
    \equiv
    \left.\frac{\partial f_3}{\partial \theta}\right|_{*}\de\theta
    +\left.\frac{\partial f_3}{\partial \dot x}\right|_{*}\de\dot x
    +\left.\frac{\partial f_3}{\partial \dot y}\right|_{*}\de\dot y
    +\left.\frac{\partial f_3}{\partial \dot\theta}\right|_{*}\de\dot\theta.
\end{split}
\end{equation}
Expanding the left-hand side in the independent variations yields
\begin{equation}\label{Thm1_identity_03}
\begin{split}
    &a_{33}(t)\de\dot x
    +a_{34}(t)\de\dot y 
    +\Bigl[a_{35}(t)-a_{33}(t)r_c\cos\bigl(\theta^*(t)\bigr)-a_{34}(t)r_c\sin\bigl(\theta^*(t)\bigr)\Bigr]\de\dot\theta
    +\bigl[\dots\bigr]\de\theta \\
    &\qquad\equiv
    \left.\frac{\partial f_3}{\partial \theta}\right|_{*}\de\theta
    +\left.\frac{\partial f_3}{\partial \dot x}\right|_{*}\de\dot x
    +\left.\frac{\partial f_3}{\partial \dot y}\right|_{*}\de\dot y
    +\left.\frac{\partial f_3}{\partial \dot\theta}\right|_{*}\de\dot\theta.
\end{split}
\end{equation}
Matching the coefficients of the independent variations $\de\dot x$, $\de\dot y$, and $\de\dot\theta$, we obtain
\begin{eqnarray}
    a_{33}(t)
    &=&
    \left.\frac{\partial f_3}{\partial \dot x}\right|_{*}
    =
    -\cos\bigl(\theta^*(t)\bigr)\sin\bigl(\theta^*(t)\bigr)\dot\theta^*(t),\label{Thm1:a33}\\
    a_{34}(t)
    &=&
    \left.\frac{\partial f_3}{\partial \dot y}\right|_{*}
    =
    -\sin^2\bigl(\theta^*(t)\bigr)\dot\theta^*(t),\label{Thm1:a34}\\
    a_{35}(t)
    &=&
    \left.\frac{\partial f_3}{\partial \dot\theta}\right|_{*}
    +a_{33}(t)r_c\cos\bigl(\theta^*(t)\bigr)
    +a_{34}(t)r_c\sin\bigl(\theta^*(t)\bigr)\nonumber\\
    &=&
    \Bigl[2r_c\dot\theta^*(t)
    -\sin\bigl(\theta^*(t)\bigr)\dot y^*(t)
    -\cos\bigl(\theta^*(t)\bigr)\dot x^*(t)\Bigr]
    \sin\bigl(\theta^*(t)\bigr)\nonumber\\    &&\quad
    -r_c\cos^2\bigl(\theta^*(t)\bigr)\sin\bigl(\theta^*(t)\bigr)\dot\theta^*(t)
    -r_c\sin^3\bigl(\theta^*(t)\bigr)\dot\theta^*(t)\nonumber\\
    &=&
    r_c\sin\bigl(\theta^*(t)\bigr)\dot\theta^*(t)
    \Bigl[
        2
        -\cos^2\bigl(\theta^*(t)\bigr)
        -\sin^2\bigl(\theta^*(t)\bigr)
        -1
    \Bigr]
    =0,
    \label{Thm1:a35}
\end{eqnarray}
where in the last step we used
\[
    \dot x^*(t)=r_c\cos\bigl(\theta^*(t)\bigr)\dot\theta^*(t),
    \qquad
    \dot y^*(t)=r_c\sin\bigl(\theta^*(t)\bigr)\dot\theta^*(t),
\]
together with $\cos^2(\theta^*(t))+\sin^2(\theta^*(t))=1$. Therefore,
\begin{equation}\label{Thm1:Dx3}
    \hbox{$\frac{d}{dt}$}\de x_{3_\bot}(t)
    =
    a_{33}(t)\de x_{3_\bot}(t)
    +a_{34}(t)\de x_{4_\bot}(t)
    +b_3(t)\de u(t).
\end{equation}

\medskip
\noindent
{\bf Step 4.} Consider the fourth transverse coordinate
\(
    x_{4_\bot}(X)=\dot y-r_c\sin(\theta)\dot\theta.
\)
Along a perturbed solution $\bigl[X(t),u(t)\bigr]$ of \eqref{Constrained_dynamics}, its time derivative is
\begin{align}
    \hbox{$\frac{d}{dt}$}\bigl[x_{4_\bot}(X(t))\bigr]
    &= \hbox{$\frac{d}{dt}$}\bigl[\dot y(t)-r_c\sin(\theta(t))\dot\theta(t)\bigr] 
    = \ddot y(t)-r_c\cos(\theta(t))\dot\theta(t)^2-r_c\sin(\theta(t))\ddot\theta(t) \nonumber\\
    &= \Bigl[\dot y(t)\sin\bigl(\theta(t)\bigr)+\dot x(t)\cos\bigl(\theta(t)\bigr)\Bigr]
       \dot\theta(t)\cos\bigl(\theta(t)\bigr) \nonumber\\    &\qquad
       -r_c\cos\bigl(\theta(t)\bigr)\dot\theta(t)^2
       -\frac{r_c}{J}\sin\bigl(\theta(t)\bigr)u(t)
       \triangleq f_4(X(t),u(t)).
    \label{Thm1:dx4}
\end{align}
Since $x_{4_\bot}$ vanishes on the nominal motion, we have
\begin{equation}\label{Thm1:x4}
    x_{4_\bot}(X(t))
    =
    \de x_{4_\bot}(t)+\bigl[\hbox{higher-order terms}\bigr].
\end{equation}
Next, linearize the function
\[
    f_4(X,u)
    =
    \Bigl[\dot y\sin(\theta)+\dot x\cos(\theta)\Bigr]\cos(\theta)\dot\theta
    -r_c\cos(\theta)\dot\theta^2
    -\frac{r_c}{J}\sin(\theta)u.
\]
Thus,
\begin{equation}
\begin{split}
\label{Thm1:f4}
    f_4(X(t),u(t))
    &=
    a_{41}(t)\de x_{1_\bot}
    +a_{42}(t)\de x_{2_\bot}
    +a_{43}(t)\de x_{3_\bot}
    +a_{44}(t)\de x_{4_\bot}
    +a_{45}(t)\de x_{5_\bot}
    +b_4(t)\de u \\&
    +\bigl[\hbox{higher-order terms}\bigr].
\end{split}
\end{equation}
At the same time, the first-order Taylor expansion of $f_4$ with respect to the independent state variations has the form
\begin{equation}\label{Thm1:f4_Taylor_series}
\begin{split}
    f_4(X(t),u(t))
    &=
    \left.\frac{\partial f_4}{\partial \theta}\right|_{*}\de\theta
    +\left.\frac{\partial f_4}{\partial \dot x}\right|_{*}\de\dot x
    +\left.\frac{\partial f_4}{\partial \dot y}\right|_{*}\de\dot y
    +\left.\frac{\partial f_4}{\partial \dot\theta}\right|_{*}\de\dot\theta
    +\left.\frac{\partial f_4}{\partial u}\right|_{*}\de u 
    +\bigl[\hbox{higher-order terms}\bigr],
\end{split}
\end{equation}
where
\begin{equation}\label{Thm1:paf4}
    \begin{array}{rcl}
        \displaystyle \left.\frac{\partial f_4}{\partial \dot x}\right|_{*}
        &=&
        \cos^2\bigl(\theta^*(t)\bigr)\dot\theta^*(t),\\[3mm]
        \displaystyle \left.\frac{\partial f_4}{\partial \dot y}\right|_{*}
        &=&
        \sin\bigl(\theta^*(t)\bigr)\cos\bigl(\theta^*(t)\bigr)\dot\theta^*(t),\\[3mm]
        \displaystyle \left.\frac{\partial f_4}{\partial \dot\theta}\right|_{*}
        &=&
        \Bigl[
            \sin\bigl(\theta^*(t)\bigr)\dot y^*(t)
            +\cos\bigl(\theta^*(t)\bigr)\dot x^*(t)
            -2r_c\dot\theta^*(t)
        \Bigr]\cos\bigl(\theta^*(t)\bigr),\\[3mm]
        \displaystyle \left.\frac{\partial f_4}{\partial u}\right|_{*}
        &=&
        -\frac{r_c}{J}\sin\bigl(\theta^*(t)\bigr).
    \end{array}
\end{equation}
The explicit expression for $\left.\frac{\partial f_4}{\partial \theta}\right|_{*}$ is not needed below. Equating the first-order terms in \eqref{Thm1:f4} and \eqref{Thm1:f4_Taylor_series}, we obtain
\begin{equation}\label{Thm1_identity4}
\begin{split}
    &a_{41}(t)\de x_{1_\bot}
    +a_{42}(t)\de x_{2_\bot}
    +a_{43}(t)\de x_{3_\bot}
    +a_{44}(t)\de x_{4_\bot}
    +a_{45}(t)\de x_{5_\bot}
    +b_4(t)\de u \\
    &\qquad\equiv
    \left.\frac{\partial f_4}{\partial \theta}\right|_{*}\de\theta
    +\left.\frac{\partial f_4}{\partial \dot x}\right|_{*}\de\dot x
    +\left.\frac{\partial f_4}{\partial \dot y}\right|_{*}\de\dot y
    +\left.\frac{\partial f_4}{\partial \dot\theta}\right|_{*}\de\dot\theta
    +\left.\frac{\partial f_4}{\partial u}\right|_{*}\de u.
\end{split}
\end{equation}
As in Step 3, the variations $\de x_\bot=\bigl[\de x_{1_\bot};\dots;\de x_{5_\bot}\bigr]$ on the left-hand side are related to the independent state variations by
\[
    \de x_\bot = J_\bot(X(t))
    \left[\!\begin{array}{c}
        \de q\\
        \de\dot q
    \end{array}\!\right],
\]
with the Jacobian matrix $J_\bot(X)$ defined in \eqref{Lem3:J}. Since the independent variations $\de x$ and $\de y$ do not appear on the right-hand side of \eqref{Thm1_identity4}, while they enter the left-hand side through $\de x_{1_\bot}$ and $\de x_{2_\bot}$, it follows that
\begin{equation}\label{Thm1:a41a42=0}
    a_{41}(t)\equiv 0,
    \qquad
    a_{42}(t)\equiv 0.
\end{equation}
Similarly,
\begin{equation}\label{Thm1:b4=b4}
    b_4(t)
    =
    \left.\frac{\partial f_4}{\partial u}\right|_{*}
    =
    -\frac{r_c}{J}\sin\bigl(\theta^*(t)\bigr).
\end{equation}
Substituting \eqref{Thm1:a41a42=0} and \eqref{Thm1:b4=b4} into \eqref{Thm1_identity4}, and using
\[
    \de x_{3_\bot}
    =
    \de\dot x
    +r_c\sin\bigl(\theta^*(t)\bigr)\dot\theta^*(t)\de\theta
    -r_c\cos\bigl(\theta^*(t)\bigr)\de\dot\theta,
\]\[
    \de x_{4_\bot}
    =
    \de\dot y
    -r_c\cos\bigl(\theta^*(t)\bigr)\dot\theta^*(t)\de\theta
    -r_c\sin\bigl(\theta^*(t)\bigr)\de\dot\theta,
\]
and $\de x_{5_\bot}=\de\dot\theta$, we obtain
\begin{equation}\label{Thm1_identity_04}
\begin{split}
    &a_{43}(t)\de x_{3_\bot}
    +a_{44}(t)\de x_{4_\bot}
    +a_{45}(t)\de x_{5_\bot} 
    \equiv
    \left.\frac{\partial f_4}{\partial \theta}\right|_{*}\de\theta
    +\left.\frac{\partial f_4}{\partial \dot x}\right|_{*}\de\dot x
    +\left.\frac{\partial f_4}{\partial \dot y}\right|_{*}\de\dot y
    +\left.\frac{\partial f_4}{\partial \dot\theta}\right|_{*}\de\dot\theta.
\end{split}
\end{equation}
Expanding the left-hand side in the independent variations yields
\begin{equation}\label{Thm1_identity_04a}
\begin{split}
    &a_{43}(t)\de\dot x
    +a_{44}(t)\de\dot y 
    +\Bigl[a_{45}(t)-a_{43}(t)r_c\cos\bigl(\theta^*(t)\bigr)-a_{44}(t)r_c\sin\bigl(\theta^*(t)\bigr)\Bigr]\de\dot\theta
    +\bigl[\dots\bigr]\de\theta \\
    &\qquad\equiv
    \left.\frac{\partial f_4}{\partial \theta}\right|_{*}\de\theta
    +\left.\frac{\partial f_4}{\partial \dot x}\right|_{*}\de\dot x
    +\left.\frac{\partial f_4}{\partial \dot y}\right|_{*}\de\dot y
    +\left.\frac{\partial f_4}{\partial \dot\theta}\right|_{*}\de\dot\theta.
\end{split}
\end{equation}
Matching the coefficients of the independent variations $\de\dot x$, $\de\dot y$, and $\de\dot\theta$, we obtain
\begin{eqnarray}
    \!\!\!\! \!\!\!\!a_{43}(t)
    &=&
    \left.\frac{\partial f_4}{\partial \dot x}\right|_{*}
    =
    \cos^2\bigl(\theta^*(t)\bigr)\dot\theta^*(t),\label{Thm1:a43}\\
     \!\!\!\! \!\!\!\!a_{44}(t)
    &=&
    \left.\frac{\partial f_4}{\partial \dot y}\right|_{*}
    =
    \sin\bigl(\theta^*(t)\bigr)\cos\bigl(\theta^*(t)\bigr)\dot\theta^*(t),\label{Thm1:a44}\\
     \!\!\!\! \!\!\!\!a_{45}(t)
    &=&
    \left.\frac{\partial f_4}{\partial \dot\theta}\right|_{*}
    +a_{43}(t)r_c\cos\bigl(\theta^*(t)\bigr)
    +a_{44}(t)r_c\sin\bigl(\theta^*(t)\bigr)\nonumber\\
    &=&
    \Bigl[
        \sin\bigl(\theta^*(t)\bigr)\dot y^*(t)
        +\cos\bigl(\theta^*(t)\bigr)\dot x^*(t)
        -2r_c\dot\theta^*(t)
    \Bigr]\cos\bigl(\theta^*(t)\bigr)\nonumber\\    &&\quad
    +r_c\cos^3\bigl(\theta^*(t)\bigr)\dot\theta^*(t)
    +r_c\sin^2\bigl(\theta^*(t)\bigr)\cos\bigl(\theta^*(t)\bigr)\dot\theta^*(t)\nonumber\\
    &=&
    r_c\cos\bigl(\theta^*(t)\bigr)\dot\theta^*(t)
    \Bigl[
        \sin^2\bigl(\theta^*(t)\bigr)
        +\cos^2\bigl(\theta^*(t)\bigr)
        -2
        +\cos^2\bigl(\theta^*(t)\bigr)
        +\sin^2\bigl(\theta^*(t)\bigr)
    =0,
    \label{Thm1:a45}
\end{eqnarray}
where in the last step we used
\[
    \dot x^*(t)=r_c\cos\bigl(\theta^*(t)\bigr)\dot\theta^*(t),
    \qquad
    \dot y^*(t)=r_c\sin\bigl(\theta^*(t)\bigr)\dot\theta^*(t),
\]
together with $\cos^2(\theta^*(t))+\sin^2(\theta^*(t))=1$. Therefore,
\begin{equation}\label{Thm1:Dx4}
    \hbox{$\frac{d}{dt}$}\de x_{4_\bot}(t)
    =
    a_{43}(t)\de x_{3_\bot}(t)
    +a_{44}(t)\de x_{4_\bot}(t)
    +b_4(t)\de u(t).
\end{equation}

\medskip
\noindent
{\bf Step 5.} Consider the fifth transverse coordinate
\(
    x_{5_\bot}(X)=\dot\theta-\om_0^*.
\)
Along a perturbed solution $\bigl[X(t),u(t)\bigr]$ of \eqref{Constrained_dynamics}, its time derivative is
\begin{equation}\label{Thm1:dx5}
    \hbox{$\frac{d}{dt}$}\bigl[x_{5_\bot}(X(t))\bigr]
    =
    \hbox{$\frac{d}{dt}$}\bigl[\dot\theta(t)-\om_0^*\bigr]
    =
    \ddot\theta(t)
    =
    \frac{1}{J}u(t)
    \triangleq f_5(X(t),u(t)).
\end{equation}
Since $x_{5_\bot}$ vanishes on the nominal motion, we have
\[
    x_{5_\bot}(X(t))
    =
    \de x_{5_\bot}(t)+\bigl[\hbox{higher-order terms}\bigr].
\]
Moreover, the function
\[
    f_5(X,u)=\frac{1}{J}u
\]
does not depend on the state variables. Hence, its first-order Taylor expansion at the nominal motion is simply
\[
    f_5(X(t),u(t))
    =
    \frac{1}{J}\de u(t)+\bigl[\hbox{higher-order terms}\bigr].
\]
Substituting these relations into \eqref{Thm1:dx5}, we obtain
\[
    \hbox{$\frac{d}{dt}$}\Bigl[\de x_{5_\bot}(t)+\bigl[\hbox{higher-order terms}\bigr]\Bigr]
    =
    \frac{1}{J}\de u(t)+\bigl[\hbox{higher-order terms}\bigr].
\]
Therefore,
\begin{equation}\label{Thm1:Dx5}
    \hbox{$\frac{d}{dt}$}\de x_{5_\bot}(t)=\frac{1}{J}\de u(t),
\end{equation}
which gives the fifth row of \eqref{Transverse_linearization}. Thus,
\(
    b_5(t)\equiv \frac{1}{J}.
\)

Collecting the relations obtained in Steps 1--5, namely \eqref{Thm1:Dx1}, \eqref{Thm1:Dx2}, \eqref{Thm1:a31a32=0}, \eqref{Thm1:b3=b3}, \eqref{Thm1:a33}--\eqref{Thm1:a35}, \eqref{Thm1:a41a42=0}, \eqref{Thm1:b4=b4}, \eqref{Thm1:a43}--\eqref{Thm1:a45}, and \eqref{Thm1:Dx5}, we obtain the variational dynamics \eqref{Transverse_linearization}.
\end{proof}

%

\section{Proof of Lemma~\ref{Lem4:invariant_subspace_of_TL}} 
\begin{proof}
We prove the three statements in order.

\medskip
\noindent
{\em Proof of statement 1.}
Consider the scalar function
\[
    I\bigl(t,\de X_\bot(t)\bigr)
    \triangleq
    \de x_{4_\bot}(t)\cos\bigl(\theta^*(t)\bigr)
    -\de x_{3_\bot}(t)\sin\bigl(\theta^*(t)\bigr).
\]
Differentiate it along a solution of \eqref{Transverse_linearization}:
\begin{align}
    \hbox{$\frac{d}{dt}$}I
    &=
    \hbox{$\frac{d}{dt}$}\bigl[\de x_{4_\bot}(t)\bigr]\cos\bigl(\theta^*(t)\bigr)
    +\de x_{4_\bot}(t)\hbox{$\frac{d}{dt}$}\bigl[\cos\bigl(\theta^*(t)\bigr)\bigr] \nonumber\\
    &\qquad
    -\hbox{$\frac{d}{dt}$}\bigl[\de x_{3_\bot}(t)\bigr]\sin\bigl(\theta^*(t)\bigr)
    -\de x_{3_\bot}(t)\hbox{$\frac{d}{dt}$}\bigl[\sin\bigl(\theta^*(t)\bigr)\bigr] \nonumber\\
    &=
    \bigl[a_{43}(t)\de x_{3_\bot}(t)+a_{44}(t)\de x_{4_\bot}(t)+b_4(t)\de u(t)\bigr]
    \cos\bigl(\theta^*(t)\bigr) \nonumber\\
    &\qquad
    -\de x_{4_\bot}(t)\sin\bigl(\theta^*(t)\bigr)\dot\theta^*(t) \nonumber\\
    &\qquad
    -\bigl[a_{33}(t)\de x_{3_\bot}(t)+a_{34}(t)\de x_{4_\bot}(t)+b_3(t)\de u(t)\bigr]
    \sin\bigl(\theta^*(t)\bigr) \nonumber\\
    &\qquad
    -\de x_{3_\bot}(t)\cos\bigl(\theta^*(t)\bigr)\dot\theta^*(t) \nonumber\\
    &=
    \varepsilon_3(t)\de x_{3_\bot}(t)
    +\varepsilon_4(t)\de x_{4_\bot}(t)
    +\varepsilon_u(t)\de u(t),
    \label{Lem4:dI}
\end{align}
where
\begin{align*}
    \varepsilon_3(t)
    &=
    a_{43}(t)\cos\bigl(\theta^*(t)\bigr)
    -a_{33}(t)\sin\bigl(\theta^*(t)\bigr)
    -\om_0^*\cos\bigl(\theta^*(t)\bigr),\\
    \varepsilon_4(t)
    &=
    a_{44}(t)\cos\bigl(\theta^*(t)\bigr)
    -a_{34}(t)\sin\bigl(\theta^*(t)\bigr)
    -\om_0^*\sin\bigl(\theta^*(t)\bigr),\\
    \varepsilon_u(t)
    &=
    b_4(t)\cos\bigl(\theta^*(t)\bigr)
    -b_3(t)\sin\bigl(\theta^*(t)\bigr).
\end{align*}
Now substitute the coefficients from \eqref{Transverse_linearization}:
\begin{align*}
    \varepsilon_3(t)
    &=
    \Bigl[\om_0^*\cos^2\bigl(\theta^*(t)\bigr)\Bigr]\cos\bigl(\theta^*(t)\bigr)
    -\Bigl[-\om_0^*\sin\bigl(\theta^*(t)\bigr)\cos\bigl(\theta^*(t)\bigr)\Bigr]\sin\bigl(\theta^*(t)\bigr) 
    -\om_0^*\cos\bigl(\theta^*(t)\bigr) \\
    &=
    \om_0^*\cos\bigl(\theta^*(t)\bigr)
    \Bigl[\cos^2\bigl(\theta^*(t)\bigr)+\sin^2\bigl(\theta^*(t)\bigr)-1\Bigr]
    \equiv 0,\\[2mm]
    \varepsilon_4(t)
    &=
    \Bigl[\om_0^*\sin\bigl(\theta^*(t)\bigr)\cos\bigl(\theta^*(t)\bigr)\Bigr]\cos\bigl(\theta^*(t)\bigr)
    -\Bigl[-\om_0^*\sin^2\bigl(\theta^*(t)\bigr)\Bigr]\sin\bigl(\theta^*(t)\bigr)
    -\om_0^*\sin\bigl(\theta^*(t)\bigr) \\
    &=
    \om_0^*\sin\bigl(\theta^*(t)\bigr)
    \Bigl[\cos^2\bigl(\theta^*(t)\bigr)+\sin^2\bigl(\theta^*(t)\bigr)-1\Bigr]
    \equiv 0,\\[2mm]
    \varepsilon_u(t)
    &=
    \Bigl[-\frac{r_c}{J}\sin\bigl(\theta^*(t)\bigr)\Bigr]\cos\bigl(\theta^*(t)\bigr)
    -\Bigl[-\frac{r_c}{J}\cos\bigl(\theta^*(t)\bigr)\Bigr]\sin\bigl(\theta^*(t)\bigr)
    \equiv 0.
\end{align*}
Hence, the right-hand side of \eqref{Lem4:dI} is identically zero, that is,
\[
    \hbox{$\frac{d}{dt}$}I\bigl(t,\de X_\bot(t)\bigr)\equiv 0
    \qquad \forall t\ge 0.
\]
Therefore, the value of $I(\cdot)$ remains constant along every solution, independently of the control input $\de u(t)$.

\medskip
\noindent
{\em Proof of statement 2.}
Assume that the initial condition satisfies
\[
    I\bigl(0,\de X_\bot(0)\bigr)
    =
    \de x_{4_\bot}(0)\cos\bigl(\theta^*(0)\bigr)
    -\de x_{3_\bot}(0)\sin\bigl(\theta^*(0)\bigr)
    \neq 0.
\]
Then the corresponding solution of \eqref{Transverse_linearization} cannot converge to the origin, because statement 1 implies that $I\bigl(t,\de X_\bot(t)\bigr)$ is constant for all $t\ge 0$. If the solution converged to the origin, then necessarily
$I\bigl(t,\de X_\bot(t)\bigr)\to 0 \hbox{  as  } t\to+\infty,$ 
which contradicts the fact that this quantity is constant and nonzero at \(t=0\). Therefore, the origin cannot be asymptotically stabilized by any control input, and the linear control system \eqref{Transverse_linearization} is not stabilizable.

\medskip
\noindent
{\em Proof of statement 3.}
To prove the last statement, observe that the vector function $\vv{\varphi}_4(t)$ defined in \eqref{Lem5:z4} is a solution of the homogeneous linear system \eqref{Lem4:linear_system}. Since $\vv{\varphi}_4(t)$ contains terms growing linearly in $t$, this solution is unbounded. Hence, the origin of \eqref{Lem4:linear_system} is unstable, and the system possesses at least a one-dimensional subspace of unbounded solutions.
\end{proof}

\section{Proof of Theorem~\ref{Thm4:main_result} and related claims} 
\begin{proof}
Proof of the main result relies on the following technical statement, whose proof is given later in the Appendix.

\begin{statement}\label{Thm2:Z_transverse coordinates}
The state transformation \(T(X)\) for the nonlinear control system \eqref{Constrained_dynamics},
\begin{equation}\label{Thm2:X_to_Z1}
    X \triangleq \bigl[\theta;x;y;\dot\theta;\dot x;\dot y\bigr]
    \quad\xrightarrow{\;T\;}\quad
    \bigl[\theta;Z\bigr],\qquad \hbox{where}
\end{equation}
\begin{equation}\label{Thm2:X_to_Z2}
    Z \triangleq \bigl[z_1;z_2;z_3;z_4;z_5\bigr]
    =
    D\left[\Phi\!\left(\frac{\theta-\theta_0^*}{\om_0^*}\right)\right]^{-1}X_\bot(X),
\end{equation}
\[    
X_\bot(X)\triangleq \bigl[x_{1_\bot}(X);x_{2_\bot}(X);x_{3_\bot}(X);x_{4_\bot}(X);x_{5_\bot}(X)\bigr],
\]
with \(D\), \(\Phi(\cdot)\), and \(X_\bot(X)\) introduced in \eqref{Lem6:change_of_coordinates}, \eqref{Lem5:Phi}, and \eqref{Transverse_coordinates}, transforms the dynamics of \eqref{Constrained_dynamics} into
\begin{eqnarray}
    \hbox{$\frac{d}{dt}$}\theta
    &=&
    \om_0^*+\frac{1}{J}z_3+z_5(0), \label{Thm3:clp_system_theta_Z123}\\
    \hbox{$\frac{d}{dt}$}
    \left[\begin{array}{c}
        z_1\\[2mm]
        z_2\\[2mm]
        z_3
    \end{array}\right]
    &=&
    \left[\begin{array}{c}
        \phantom{-}\frac{1}{J}\cos(\theta)\,z_3^2\\[2mm]
        -\frac{1}{J}\sin(\theta)\,z_3^2\\[2mm]
        0
    \end{array}\right]
    +
    \left[\begin{array}{c}
        \phantom{-}\cos(\theta)\,z_3\\[2mm]
        -\sin(\theta)\,z_3\\[2mm]
        0
    \end{array}\right]z_5(0)
    +
    \left[\begin{array}{c}
        \sin(\theta)\\[2mm]
        \cos(\theta)\\[2mm]
        1
    \end{array}\right]u, \label{Thm3:clp_system_Z123}\\
    \hbox{$\frac{d}{dt}$}
    \left[\begin{array}{c}
        z_4\\
        z_5
    \end{array}\right]
    &=&
    \left[\begin{array}{c}
        0\\
        0
    \end{array}\right].
    \label{Thm2:Z2_dynamics}
\end{eqnarray}
Furthermore, the determinant of the Jacobian \(J_T(X)\) of the nonlinear state transformation \(T(X)\) defined by \eqref{Thm2:X_to_Z1}--\eqref{Thm2:X_to_Z2} is constant and nonzero on the nominal motion \(q^*(t)\).
\end{statement}

To prove Theorem~\ref{Thm4:main_result}, we first rewrite the dynamics of the closed-loop system \eqref{Constrained_dynamics}, \eqref{Thm4:nonlinear_feedback} in the coordinates \((\theta;Z)\) introduced by \eqref{Thm2:X_to_Z1}--\eqref{Thm2:X_to_Z2}.
By Claim~\ref{Thm2:Z_transverse coordinates}, the closed-loop system takes the form
\begin{eqnarray}
    \hbox{$\frac{d}{dt}$}\theta
    &=&
    \om_0^*+\frac{1}{J}z_3+z_5(0),
    \label{P:Thm4:theta_dynamics}\\
    \hbox{$\frac{d}{dt}$}
z
    &=&
    F_0(\theta,Z)+F_1(\theta,Z),
    \label{P:Thm4:Z1_dynamics}\\
    \hbox{$\frac{d}{dt}$}
    \left[\begin{array}{c}
        z_4\\
        z_5
    \end{array}\right]
    &=&
    \left[\begin{array}{c}
        0\\
        0
    \end{array}\right],
    \label{P:Thm4:Z2_dynamics}
\end{eqnarray}
where \quad\(z=\bigl[z_1;~z_2;~z_3\bigr]\),\quad while 
\[
    F_0(\theta,Z)\triangleq
    \left[\!\begin{array}{c}
        \phantom{-}\frac{1}{J}\cos\bigl(\theta\bigr)\,z_3\!\left(\frac{1}{J}z_3+z_5(0)\right)\\[2mm]
        -\frac{1}{J}\sin\bigl(\theta\bigr)\,z_3\!\left(\frac{1}{J}z_3+z_5(0)\right)\\[2mm]
        0
    \end{array}\!\right]=
        \left[\!\begin{array}{c}
        f_{01}(\theta;Z)\\[2mm]
        f_{02}(\theta;Z)\\[2mm]
        f_{03}(\theta;Z)
    \end{array}\!\right],
	\]\[
    F_1(\theta,Z)\triangleq
    \left[\!\begin{array}{c}
        \sin\bigl(\theta\bigr)\\[2mm]
        \cos\bigl(\theta\bigr)\\[2mm]
        1
    \end{array}\!\right]U\bigl(T^{-1}(\theta,Z)\bigr)=
        \left[\!\begin{array}{c}
        f_{11}(\theta;Z)\\[2mm]
        f_{12}(\theta;Z)\\[2mm]
        f_{13}(\theta;Z)
    \end{array}\!\right].
\]
Here the controller \(u=U(X)\) is expressed in the new coordinates \((\theta;Z)\).

By assumption, the linearization of the standard set of transverse coordinates \eqref{Transverse_coordinates} along the nominal motion \(q^*(t)\) of the closed-loop system \eqref{Constrained_dynamics}, \eqref{Thm4:nonlinear_feedback} possesses an invariant quadratically stable subspace of dimension three.

For the closed-loop system written in the form \eqref{P:Thm4:theta_dynamics}--\eqref{P:Thm4:Z2_dynamics}, this means that the linearization of subsystem \eqref{P:Thm4:Z1_dynamics},
\begin{equation}\label{P:Thm4:Z1_dynamics_linearization_01}
    \hbox{$\frac{d}{dt}$}
 \de z
    =
    A_{z_{clp}}\bigl(\theta^*(t)\bigr)
  \de z,
\end{equation}
where \quad\(\de z\triangleq \bigl[\de z_1;~\de z_2,~\de z_3\bigr]\)\quad and 
\begin{equation}\label{P:Thm4:Azclp}
    A_{z_{clp}}\bigl(\theta^*(t)\bigr)\triangleq
    \left.
    \left[\begin{array}{ccc}
        \frac{\partial}{\partial z_1}(f_{01}+f_{11}) &
        \frac{\partial}{\partial z_2}(f_{01}+f_{11}) &
        \frac{\partial}{\partial z_3}(f_{01}+f_{11}) \\[1mm]
        \frac{\partial}{\partial z_1}(f_{02}+f_{12}) &
        \frac{\partial}{\partial z_2}(f_{02}+f_{12}) &
        \frac{\partial}{\partial z_3}(f_{02}+f_{12}) \\[1mm]
        \frac{\partial}{\partial z_1}(f_{03}+f_{13}) &
        \frac{\partial}{\partial z_2}(f_{03}+f_{13}) &
        \frac{\partial}{\partial z_3}(f_{03}+f_{13})
    \end{array}\right]
    \right|_{\begin{array}{rcl}
        \theta&=&\theta^*(t)\\[-0.5mm]
        Z&=&0
    \end{array}}
\end{equation}
is quadratically stable.

Now note that every component of \(F_0(\theta,Z)\) contains the factor \(z_3\). In particular, along the nominal motion one has \(Z=0\), and therefore \(z_5(0)=0\). Hence,
\[
\left.
\left[\begin{array}{ccc}
    \frac{\partial f_{01}}{\partial z_1} & \frac{\partial f_{01}}{\partial z_2} & \frac{\partial f_{01}}{\partial z_3}\\[1mm]
    \frac{\partial f_{02}}{\partial z_1} & \frac{\partial f_{02}}{\partial z_2} & \frac{\partial f_{02}}{\partial z_3}\\[1mm]
    \frac{\partial f_{03}}{\partial z_1} & \frac{\partial f_{03}}{\partial z_2} & \frac{\partial f_{03}}{\partial z_3}
\end{array}\right]
\right|_{\begin{array}{rcl}
    \theta\!\!&\!\!=\!\!&\!\!\theta^*(t)\\[-0.5mm]
    Z\!\!&\!\!=\!\!&\!\!0
\end{array}}
=
\left[\begin{array}{ccc}
    0 & 0 & 0\\
    0 & 0 & 0\\
    0 & 0 & 0
\end{array}\right].
\]

Furthermore, since \(U(\cdot)\) is smooth and vanishes on the nominal motion, we obtain
\begin{align}
    A_{z_{clp}}\bigl(\theta^*(t)\bigr)
    &=
    \left.
    \left\{
    \left[\!\!\begin{array}{c}
        \sin\bigl(\theta\bigr)\\
        \cos\bigl(\theta\bigr)\\
        1
    \end{array}\!\!\right]
    \left[
        \frac{\partial U(T^{-1}(\theta,Z))}{\partial z_1},
        \frac{\partial U(T^{-1}(\theta,Z))}{\partial z_2},
        \frac{\partial U(T^{-1}(\theta,Z))}{\partial z_3}
    \right]
    \right\}
    \right|_{\begin{array}{rcl}
        \theta\!\!&\!\!=\!\!&\!\!\theta^*(t)\\[-0.5mm]
        Z\!\!&\!\!=\!\!&\!\!0
    \end{array}}
    \nonumber\\
    &=
    \left[\!\!\begin{array}{c}
        \sin\bigl(\theta^*(t)\bigr)\\
        \cos\bigl(\theta^*(t)\bigr)\\
        1
    \end{array}\!\!\right]
    \Bigl[K_1\bigl(\theta^*(t)\bigr),\,K_2\bigl(\theta^*(t)\bigr),\,K_3\bigl(\theta^*(t)\bigr)\Bigr].
    \label{P:Thm4:Z1_Az}
\end{align}
Hence, the matrix function \(A_{z_{clp}}\bigl(\theta^*(t)\bigr)\) is continuous and \(T_0\)-periodic, where \(T_0=\frac{2\pi}{\om_0^*}\).


Quadratic stability of the \(T_0\)-periodic linear system \eqref{P:Thm4:Z1_dynamics_linearization_01} implies the existence of a Floquet--Lyapunov transformation \cite{Yakubovich}
\begin{equation}\label{Thm4:Floquet-Lyapunov_transformation}
    L\bigl(\theta^*(t)\bigr)\de W
    =
    \de z,
    \qquad
    L\bigl(\theta^*(t)\bigr)=L\bigl(\theta^*(t+T_0)\bigr),
    \qquad
    \det L\bigl(\theta^*(t)\bigr)\neq 0,\ \forall t,
\end{equation}
such that the linear system \eqref{P:Thm4:Z1_dynamics_linearization_01}, expressed in the coordinates
\(
    \de W=\bigl[\de w_1;\de w_2;\de w_3\bigr],
\)
becomes time-invariant:
\begin{equation}\label{P:Thm4:Z1_dL_eqn}
    \hbox{$\frac{d}{dt}$}\de W
    =
    \left[L\bigl(\theta^*(t)\bigr)\right]^{-1}
    \Bigl[
        A_{z_{clp}}\bigl(\theta^*(t)\bigr)L\bigl(\theta^*(t)\bigr)
        -
        \hbox{$\frac{d}{dt}$}L\bigl(\theta^*(t)\bigr)
    \Bigr]\de W
    =
    R_0\de W,
\end{equation}
where \(R_0\) is a strictly Hurwitz matrix.

Comparing the left- and right-hand sides of \eqref{P:Thm4:Z1_dL_eqn}, we obtain the differential equation
\begin{equation}\label{P:Thm4:Z1_dL_eqn_2}
    \hbox{$\frac{d}{dt}$}L\bigl(\theta^*(t)\bigr)
    =
    A_{z_{clp}}\bigl(\theta^*(t)\bigr)L\bigl(\theta^*(t)\bigr)
    -
    L\bigl(\theta^*(t)\bigr)R_0.
\end{equation}
For later use, it is convenient to rewrite this relation in the form
\begin{equation}\label{P:Thm4:Z1_dL_eqn_3}
    \hbox{$\frac{d}{d\tau}$}L(\tau)\,\om_0^*
    =
    A_{z_{clp}}(\tau)L(\tau)-L(\tau)R_0,
    \qquad
    \tau=\theta^*(t).
\end{equation}

Finally, along any perturbed motion of the Dubins car, the variables \(z_4(t)\equiv 0\) and \(z_5(t)\equiv z_5(0)\) are constant; see \eqref{Thm2:Z2_dynamics}. Therefore, in order to establish orbital stability of \(q^*(t)\), it remains to analyze the evolution of \(z_1(t)\), \(z_2(t)\), and \(z_3(t)\) along solutions of the closed-loop system \eqref{Thm4:nonlinear_feedback}, \eqref{Thm3:clp_system_theta_Z123}, and \eqref{Thm3:clp_system_Z123}, or equivalently \eqref{P:Thm4:theta_dynamics}--\eqref{P:Thm4:Z2_dynamics}.


To this end, let us consider the closed-loop system \eqref{P:Thm4:theta_dynamics}--\eqref{P:Thm4:Z2_dynamics} in the new variables
\begin{equation}\label{P:Thm4:Tw_transform}
    \bigl[\theta;z_1;z_2;z_3;z_4;z_5\bigr]
    \xrightarrow{\;T_w\;}
    \bigl[\theta;W;z_4;z_5\bigr],
    \qquad
    W\triangleq \bigl[w_1;w_2;w_3\bigr],
    \qquad
    z\triangleq \bigl[z_1;z_2;z_3\bigr]=L(\theta)W,
\end{equation}
where the matrix function \(L(\cdot)\) is the Floquet--Lyapunov transformation introduced in \eqref{Thm4:Floquet-Lyapunov_transformation}.

In the new coordinates, \eqref{P:Thm4:theta_dynamics} becomes
\begin{equation}\label{P:Thm4:theta_dynamics_W}
    \hbox{$\frac{d}{dt}$}\theta
    =
    \om_0^*+\ga(\theta)W+z_5(0),
    \qquad
    \ga(\theta)\triangleq \frac{1}{J}\bigl[0,\,0,\,1\bigr]L(\theta).
\end{equation}

Next, differentiating the relation \(z=L(\theta)W\), we obtain
\[
    \hbox{$\frac{d}{dt}$}z
    =
    \hbox{$\frac{d}{dt}$}L(\theta)\,W
    +L(\theta)\hbox{$\frac{d}{dt}$}W.
\]
Using \eqref{P:Thm4:Z1_dynamics}, this gives
\begin{equation}\label{P:Thm4:W_dynamics_sum}
    \hbox{$\frac{d}{dt}$}W
    =
    \bigl[L(\theta)\bigr]^{-1}
    \Bigl[
        F_0(\theta;Z)+F_1(\theta;Z)
        -
        \hbox{$\frac{d}{dt}$}L(\theta)\,W
    \Bigr].
\end{equation}
To rewrite the last term, observe that by \eqref{P:Thm4:Z1_dL_eqn_3},
\(
    \hbox{$\frac{d}{d\theta}$}L(\theta)
    =
    \frac{1}{\om_0^*}\Bigl[A_{z_{clp}}(\theta)L(\theta)-L(\theta)R_0\Bigr].
\)
Combining this identity with \eqref{P:Thm4:theta_dynamics_W}, we obtain
\[
    \hbox{$\frac{d}{dt}$}L(\theta)
    =
    \frac{1}{\om_0^*}\Bigl[A_{z_{clp}}(\theta)L(\theta)-L(\theta)R_0\Bigr]
    \bigl[\om_0^*+\ga(\theta)W+z_5(0)\bigr].
\]
Substituting this into \eqref{P:Thm4:W_dynamics_sum} yields
\begin{align}
    \hbox{$\frac{d}{dt}$}W
    &=
    \bigl[L(\theta)\bigr]^{-1}
    \Bigl[
        F_0(\theta;Z)+F_1(\theta;Z)
        -
        \frac{1}{\om_0^*}\bigl[A_{z_{clp}}(\theta)L(\theta)-L(\theta)R_0\bigr]
        \bigl[\om_0^*+\ga(\theta)W+z_5(0)\bigr]W
    \Bigr]
    \nonumber\\
    &=
    \bigl[L(\theta)\bigr]^{-1}F_0(\theta;Z)
    +
    \bigl[L(\theta)\bigr]^{-1}\Bigl[F_1(\theta;Z)-A_{z_{clp}}(\theta)L(\theta)W\Bigr]
    +
    R_0W
    \nonumber\\
    &\qquad
    -
    \bigl[L(\theta)\bigr]^{-1}
    \frac{1}{\om_0^*}
    \bigl[A_{z_{clp}}(\theta)L(\theta)-L(\theta)R_0\bigr]
    \bigl[\ga(\theta)W+z_5(0)\bigr]W.
    \label{P:Thm4:W_dynamics}
\end{align}

We now expand the three nontrivial summands on the right-hand side of \eqref{P:Thm4:W_dynamics}.

First, by \eqref{P:Thm4:Z1_dynamics},
\begin{align}
    \bigl[L(\theta)\bigr]^{-1}F_0(\theta;Z)
    &=
    z_5(0)\,
    \bigl[L(\theta)\bigr]^{-1}
    \left[\begin{array}{ccc}
        0 & 0 & \frac{1}{J}\cos(\theta)\\
        0 & 0 & -\frac{1}{J}\sin(\theta)\\
        0 & 0 & 0
    \end{array}\right]
  \left[\begin{array}{c}
        z_1\\[2mm]
        z_2\\[2mm]
        z_3
    \end{array}\right]
    +
    \biggl[\begin{array}{c}
        \hbox{higher-order }\\
        \hbox{terms in } z_3
    \end{array}\biggr]
    \nonumber\\
    &=
    z_5(0)\,
    \bigl[L(\theta)\bigr]^{-1}
    \left[\begin{array}{ccc}
        0 & 0 & \frac{1}{J}\cos(\theta)\\
        0 & 0 & -\frac{1}{J}\sin(\theta)\\
        0 & 0 & 0
    \end{array}\right]
    L(\theta)W
    +
    \biggl[\begin{array}{c}
        \hbox{higher-order }\\
        \hbox{terms in } W
    \end{array}\biggr]
    \nonumber\\    &
    =
    z_5(0)\,R_{11}(\theta)W
    +
    \biggl[\begin{array}{c}
        \hbox{higher-order}\\
        \hbox{ terms in } W
    \end{array}\biggr],
    \label{P:Thm4:LinvF0}
\end{align}
\begin{equation}\label{P:Thm4:R11}
\hbox{where}\qquad \qquad   R_{11}(\theta)\triangleq
    \bigl[L(\theta)\bigr]^{-1}
    \left[\begin{array}{ccc}
        0 & 0 & \frac{1}{J}\cos(\theta)\\
        0 & 0 & -\frac{1}{J}\sin(\theta)\\
        0 & 0 & 0
    \end{array}\right]
    L(\theta).
\end{equation}

Second,
\begin{equation}\begin{array}{l}
    \bigl[L(\theta)\bigr]^{-1}\Bigl[F_1(\theta;Z)-A_{z_{clp}}(\theta)L(\theta)W\Bigr]
		\\
    =
    \bigl[L(\theta)\bigr]^{-1}
    \left[
        \left[\begin{array}{c}
            \sin(\theta)\\
            \cos(\theta)\\
            1
        \end{array}\right]U(\theta;Z)
        -
        \left[\begin{array}{c}
            \sin(\theta)\\
            \cos(\theta)\\
            1
        \end{array}\right]
        \left.
        \left[
            \frac{\partial U(\theta;Z)}{\partial z_1},
            \frac{\partial U(\theta;Z)}{\partial z_2},
            \frac{\partial U(\theta;Z)}{\partial z_3}
        \right]
        \right|_{Z=0}
        z
    \right]
		\\
    =
    \bigl[L(\theta)\bigr]^{-1}
    \left[\begin{array}{c}
        \sin(\theta)\\
        \cos(\theta)\\
        1
    \end{array}\right]
    \Biggl(
        U(\theta;Z)
        -
        \left.
        \left[
            \frac{\partial U(\theta;Z)}{\partial z_1},
            \frac{\partial U(\theta;Z)}{\partial z_2},
            \frac{\partial U(\theta;Z)}{\partial z_3}
        \right]
        \right|_{Z=0}
        z
    \Biggr)
		\\
    =
    0_{3\times 3}\,W
    +
    \biggl[\begin{array}{c}
        \hbox{higher-order terms}\\
        \hbox{in } W
    \end{array}\biggr].
    \label{P:Thm4:LinvF1-ALW}
\end{array}\end{equation}

Third,
\begin{align}
    \bigl[L(\theta)\bigr]^{-1}
    \frac{1}{\om_0^*}
    \bigl[A_{z_{clp}}(\theta)L(\theta)-L(\theta)R_0\bigr]
    \bigl[\ga(\theta)W+z_5(0)\bigr]W
    &=
    z_5(0)\,
    \bigl[L(\theta)\bigr]^{-1}
    \frac{1}{\om_0^*}
    \bigl[A_{z_{clp}}(\theta)L(\theta)
    \nonumber\\    &    
    -L(\theta)R_0\bigr]W+
    \biggl[\begin{array}{c}
        \hbox{higher-order}\\
        \hbox{ terms in } W
    \end{array}\biggr]
    \nonumber
    \\    &
    =
    z_5(0)\,R_{13}(\theta)W
    +
    \biggl[\begin{array}{c}
        \hbox{higher-order}\\
        \hbox{ terms in } W
    \end{array}\biggr],
    \label{P:Thm4:3rd_term}
\end{align}
\begin{equation}\label{P:Thm4:R13}
\hbox{where}\qquad\qquad    R_{13}(\theta)\triangleq
    \bigl[L(\theta)\bigr]^{-1}
    \frac{1}{\om_0^*}
    \bigl[A_{z_{clp}}(\theta)L(\theta)-L(\theta)R_0\bigr].
\end{equation}

Substituting \eqref{P:Thm4:LinvF0}, \eqref{P:Thm4:LinvF1-ALW}, and \eqref{P:Thm4:3rd_term} into \eqref{P:Thm4:W_dynamics}, we obtain
\begin{equation}\label{P:Thm4:W_dynamics2}
    \hbox{$\frac{d}{dt}$}W
    =
    \Bigl[R_0+z_5(0)R_1(\theta)\Bigr]W
    +
    \biggl[\begin{array}{c}
        \hbox{higher-order terms}\\
        \hbox{in } W
    \end{array}\biggr],
\end{equation}
where
\begin{equation}\label{P:Thm4:R1}
    R_1(\theta)\triangleq R_{11}(\theta)-R_{13}(\theta).
\end{equation}
The matrix functions \(R_{11}(\theta)\) and \(R_{13}(\theta)\) are continuous and globally bounded in \(\theta\). Hence, \(R_1(\theta)\) is continuous and globally bounded as well.

Since \(R_0\) is strictly Hurwitz, there exists a symmetric positive definite matrix \(P\in\mathbb{R}^{3\times 3}\) such that
\begin{equation}\label{P:Thm4:Z1_P}
    P=P^{\trn}>0,
    \qquad
    R_0^{\trn}P+PR_0=-I_{3\times 3}.
\end{equation}
Consider the quadratic form
\[
    V(W)\triangleq W^{\trn}PW.
\]
Along a solution of \eqref{P:Thm4:theta_dynamics_W} and \eqref{P:Thm4:W_dynamics2}, its time derivative is
\begin{align*}
    \hbox{$\frac{d}{dt}$}V\bigl(W(t)\bigr)
    &=
    \left[\hbox{$\frac{d}{dt}$}W(t)\right]^{\trn}PW(t)
    +
    W(t)^{\trn}P\left[\hbox{$\frac{d}{dt}$}W(t)\right] \\
    &=
    W(t)^{\trn}
    \Bigl[
        \bigl(R_0+z_5(0)R_1(\theta)\bigr)^{\trn}P
        +
        P\bigl(R_0+z_5(0)R_1(\theta)\bigr)
    \Bigr]
    W(t)
    +
    \left[\!\!\begin{array}{c}
        \hbox{higher-order}\\
        \hbox{terms}
    \end{array}\!\!\right] \\
    &=
    -\|W(t)\|^2
    +
    z_5(0)\,W(t)^{\trn}Q(\theta)W(t)
    +
    \left[\!\!\begin{array}{c}
        \hbox{higher-order}\\
        \hbox{terms}
    \end{array}\!\!\right],
\end{align*}
where \(Q(\theta)\) is a smooth bounded matrix function.

Therefore, for a solution of the closed-loop system with sufficiently small \(\|W(0)\|\) and \(|z_5(0)|\), the last equality implies
\[
    \hbox{$\frac{d}{dt}$}\bigl[W(t)^{\trn}PW(t)\bigr]
    \le
    -\frac{1}{2}\|W(t)\|^2,
    \qquad
    \forall\,t\ge 0,
\]
which in turn implies exponential convergence of \(W(t)\) to zero.

\begin{statement}\label{LemCC:z5=}
Consider the change of coordinates \eqref{Thm2:X_to_Z1}--\eqref{Thm2:X_to_Z2} for the nonlinear control system \eqref{Constrained_dynamics}. Then the fifth component \(z_5(X(t))\) of the new transverse coordinates satisfies
\begin{equation}\label{LemCC:z5=C}
    z_5(X(t))\equiv \frac{1}{r_c}v_0-\om_0^*,
    \qquad
    v_0=\sqrt{\dot x_0^2+\dot y_0^2},
\end{equation}
and remains constant along every motion of the Dubins car.
\end{statement}

If, in addition, \eqref{Thm4:v0=Om*rc} holds, then \(z_5(0)=0\). In that case, exponential convergence of \(W(t)\) to zero implies exponential convergence of \(z(t)=L(\theta(t))W(t)\) to zero. Since \(z_4(t)\equiv 0\) and \(z_5(t)\equiv 0\), it follows that
\[
    Z(t)=\bigl[z_1(t);z_2(t);z_3(t);z_4(t);z_5(t)\bigr]\to 0
    \qquad \text{as } t\to+\infty.
\]
Finally, \eqref{Thm2:X_to_Z2} implies that the standard transverse coordinates \eqref{Transverse_coordinates} also converge to zero. This completes the proof of the main result. The remaining claims used above are proved separately below.
\end{proof}


\section{Proof of Technical Claims}

\subsection{Proof of Claim~\ref{Thm2:Z_transverse coordinates}} 
\begin{proof}
The proof of Claim~\ref{Thm2:Z_transverse coordinates} relies on the following two technical results, which are stated here and proved later in the Appendix.

\begin{statement}\label{LemAA:dynamics_in_(th,Xbot)_coordinates}
The dynamics of the nonlinear control system \eqref{Constrained_dynamics} under the change of coordinates
\begin{equation}\label{LemAA:change_X_to_(th,Xbot)}
    X \triangleq \bigl[\theta;x;y;\dot\theta;\dot x;\dot y\bigr]
    \quad\xrightarrow{\;T_1\;}\quad
    \bigl[\theta;X_\bot\bigr],
\end{equation}
where
\(
    X_\bot \triangleq \bigl[x_{1_\bot};x_{2_\bot};x_{3_\bot};x_{4_\bot};x_{5_\bot}\bigr]
\)
and the functions $x_{i_\bot}(X)$, $i=1,\dots,5$, are defined in \eqref{Transverse_coordinates}, has the form
\begin{eqnarray}
    \hbox{$\frac{d}{dt}$}\theta
    &=&
    \om_0^*+x_{5_\bot},
    \label{LemAA:dth}\\
    \hbox{$\frac{d}{dt}$}X_\bot
    &=&
    A_\bot(\theta)X_\bot + B_\bot(\theta)u + F_\bot(\theta,X_\bot),
    \label{LemAA:dXbot}
\end{eqnarray}
where
\begin{equation}\label{LemAA:A_bot}
    A_\bot(\theta)=
    \left[\begin{array}{ccccc}
        0 & 0 & 1 & 0 & 0\\
        0 & 0 & 0 & 1 & 0\\
        0 & 0 & a_{33}(\theta) & a_{34}(\theta) & 0\\
        0 & 0 & a_{43}(\theta) & a_{44}(\theta) & 0\\
        0 & 0 & 0 & 0 & 0
    \end{array}\right],
    \qquad
    B_\bot(\theta)=
    \left[\begin{array}{c}
        0\\
        0\\
        b_3(\theta)\\
        b_4(\theta)\\
        b_5(\theta)
    \end{array}\right],
\end{equation}
and the higher-order terms are defined by 
\begin{equation}\label{LemAA:F_bot}
    F_\bot(\theta,X_\bot)=
    \left[\begin{array}{c}
        0\\
        0\\
        f_3(\theta,X_\bot)\\
        f_4(\theta,X_\bot)\\
        0
    \end{array}\right].
\end{equation}
Here
\begin{equation}\label{LemAA:abf}
    \begin{array}{rclrcl}
        a_{33}(\theta) &=& -\om_0^*\sin(\theta)\cos(\theta), &
        a_{34}(\theta) &=& -\om_0^*\sin^2(\theta),\\[1mm]
        a_{43}(\theta) &=& \phantom{-}\om_0^*\cos^2(\theta), &
        a_{44}(\theta) &=& \phantom{-}\om_0^*\sin(\theta)\cos(\theta),\\[1mm]
        f_3(\theta,X_\bot)
        &=&
        -x_{5_\bot}\sin(\theta)\bigl[\cos(\theta)x_{3_\bot}+\sin(\theta)x_{4_\bot}\bigr],\\[1mm] 
				f_4(\theta,X_\bot)
        &=&
        x_{5_\bot}\cos(\theta)\bigl[\cos(\theta)x_{3_\bot}+\sin(\theta)x_{4_\bot}\bigr],
    \end{array}
\end{equation}
\[
    b_3(\theta)= -\frac{r_c}{J}\cos(\theta), \qquad
    b_4(\theta)= -\frac{r_c}{J}\sin(\theta), \qquad
    b_5(\theta)\equiv \frac{1}{J}.
\]
\end{statement}

\begin{statement}\label{LemBB:on_dPhi_nonlinear}
Define
\begin{equation}\label{LemBB:Phi_bot(th)}
    \Phi_\bot(\tau)=
    \left[\begin{array}{ccccc}
        1 & 0 & \frac{1}{\om_0^*}\sin(\tau)
          & \frac{\tau}{\om_0^*}\sin(\tau)+\frac{2}{\om_0^*}\bigl(\cos(\tau)-1\bigr) & 0\\
        0 & 1 & \frac{1}{\om_0^*}\bigl(1-\cos(\tau)\bigr)
          & -\frac{\tau}{\om_0^*}\cos(\tau)+\frac{2}{\om_0^*}\sin(\tau) & 0\\
        0 & 0 & \cos(\tau)
          & \tau\cos(\tau)-\sin(\tau) & 0\\
        0 & 0 & \sin(\tau)
          & \tau\sin(\tau)+\cos(\tau) & 0\\
        0 & 0 & 0 & 0 & 1
    \end{array}\right]
    \triangleq
    \Phi(t)\Bigl|_{\,t=\frac{\tau-\theta_0^*}{\om_0^*}},
\end{equation}
where \(\Phi(t)\) is introduced in \eqref{Lem5:Phi}. Then the following properties hold:
\begin{enumerate}
    \item
    The matrix function \(\Phi_\bot(\tau)\) satisfies
    \begin{equation}\label{LemmBB:dPhi_bot(tau)}
        \om_0^*\hbox{$\frac{d}{d\tau}$}\Phi_\bot(\tau)\equiv A_\bot(\tau)\Phi_\bot(\tau),
        \qquad
        \Phi_\bot(0)=I_{5\times 5},
    \end{equation}
    where \(A_\bot(\tau)\) is defined by \eqref{LemAA:A_bot} and \eqref{LemAA:abf};

    \item
    The inverse matrix is
    \begin{equation}\label{LemmBB:invPhi_bot(tau)}
        \bigl[\Phi_\bot(\tau)\bigr]^{-1}
        =
        \left[\begin{array}{c|c}
            I_{2\times 2} & \Lambda(\tau)\\
            \hline
            0_{3\times 2} & \Sigma(\tau)
        \end{array}\right]^{-1}
        =
        \left[\begin{array}{c|c}
            I_{2\times 2} & -\Lambda(\tau)\Sigma^{-1}(\tau)\\
            \hline
            0_{3\times 2} & \Sigma^{-1}(\tau)
        \end{array}\right],
    \end{equation}
    where
    \begin{equation}\label{LemBB:LambdaSigma}
        \Lambda(\tau)=
        \left[\begin{array}{ccc}
            \frac{1}{\om_0^*}\sin(\tau)
            & \frac{\tau}{\om_0^*}\sin(\tau)+\frac{2}{\om_0^*}\bigl(\cos(\tau)-1\bigr)
            & 0\\[1mm]
            \frac{1}{\om_0^*}\bigl(1-\cos(\tau)\bigr)
            & -\frac{\tau}{\om_0^*}\cos(\tau)+\frac{2}{\om_0^*}\sin(\tau)
            & 0
        \end{array}\right],
    \end{equation}
    \begin{equation}\label{LemBB:Sigma}
        \Sigma(\tau)=
        \left[\begin{array}{ccc}
            \cos(\tau) & \tau\cos(\tau)-\sin(\tau) & 0\\
            \sin(\tau) & \tau\sin(\tau)+\cos(\tau) & 0\\
            0 & 0 & 1
        \end{array}\right],
    \end{equation}
    are the block matrices in the corresponding partition of \(\Phi_\bot(\tau)\), and
    \(
        \det\Phi_\bot(\tau)=\det\Sigma(\tau)\equiv 1.
    \)
\end{enumerate}
\end{statement}


To prove Claim~\ref{Thm2:Z_transverse coordinates}, we proceed in two steps.
First, we verify that the transformation \eqref{Thm2:X_to_Z1}--\eqref{Thm2:X_to_Z2} indeed defines a new coordinate system for the nonlinear control system \eqref{Constrained_dynamics}, that is, that it is smooth and locally smoothly invertible in a neighborhood of the nominal motion $q^*(t)$.
Second, we derive the dynamics of \eqref{Constrained_dynamics} in the new coordinates and verify that the change of coordinates \(\bigl[\theta; Z\bigr]=T(X)\) transforms the system into \eqref{Thm3:clp_system_theta_Z123}--\eqref{Thm2:Z2_dynamics}.

To this end, let us compute the Jacobian \(J_T(X)\) of the transformation \eqref{Thm2:X_to_Z1}--\eqref{Thm2:X_to_Z2}.
It is convenient to represent this transformation as the composition \(T(X)=T_2(T_1(X))\), where
\begin{equation}\label{Thm2:T1_T2}
    X = \bigl[\theta;x;y;\dot\theta;\dot x;\dot y\bigr]
    \quad\xrightarrow{\;T_1\;}\quad
    \bigl[\theta;X_\bot\bigr]
    \quad\xrightarrow{\;T_2\;}\quad
    \bigl[\theta;Z\bigr],
\end{equation}
with
\quad\(
    X_\bot = \bigl[x_{1_\bot};x_{2_\bot};x_{3_\bot};x_{4_\bot};x_{5_\bot}\bigr],
    \quad
    Z = q \bigl[z_1;z_2;z_3;z_4;z_5\bigr].
\)\quad
Here the nontrivial part of \(T_1(X)\) is given by the standard transverse coordinates \eqref{Transverse_coordinates}, and the nontrivial part of \(T_2(\theta;Z)\) is defined by \eqref{Thm2:X_to_Z2}.

The Jacobians of \(T_1(X)\) and \(T_2(\theta;Z)\) are readily computed.
Indeed, except for the trivial first row, Jacobian \(J_1(X)\) of \(T_1(X)\) was already obtained in Lemma~\ref{Lem3:Jacobian}:
\begin{equation}\label{Thm2:JT1}
    J_1(X)=
    \left[
    \begin{array}{cccccc}
        1 & 0 & 0 & 0 & 0 & 0 \\
        j_1(X) & 1 & 0 & 0 & 0 & 0 \\
        j_2(X) & 0 & 1 & 0 & 0 & 0 \\
        -j_2(X)\dot\theta & 0 & 0 & j_1(X) & 1 & 0 \\
        \phantom{-}j_1(X)\dot\theta & 0 & 0 & j_2(X) & 0 & 1 \\
        0 & 0 & 0 & 1 & 0 & 0
    \end{array}
    \right],
\end{equation}
where
\quad\(
    j_1(X)=-r_c\cos(\theta),
    \quad
    j_2(X)=-r_c\sin(\theta).
\)\quad
Its determinant is
\begin{align*}
    \!\!\det\! J_1(X)
    &=
    \det\!\!\left[\!
    \begin{array}{ccc|ccc}
        1 & 0 & 0 & 0 & 0 & 0 \\
        j_1(X) & 1 & 0 & 0 & 0 & 0 \\
        j_2(X) & 0 & 1 & 0 & 0 & 0 \\
        \hline
        -j_2(X)\dot\theta & 0 & 0 & j_1(X) & 1 & 0 \\
        \phantom{-}j_1(X)\dot\theta & 0 & 0 & j_2(X) & 0 & 1 \\
        0 & 0 & 0 & 1 & 0 & 0
    \end{array}
    \!\right]\! 
    =
    \det\!\!\left[\!\begin{array}{ccc}
        1 & 0 & 0 \\
        j_1(X) & 1 & 0 \\
        j_2(X) & 0 & 1
    \end{array}\!\right]\!
    \det\!\!\left[\!\begin{array}{ccc}
        j_1(X) & 1 & 0 \\
        j_2(X) & 0 & 1 \\
        1 & 0 & 0
    \end{array}\!\right]\!
    \equiv 1.
\end{align*}

In turn, the Jacobian \(J_2(\cdot)\) of \(T_2(\theta;Z)\) has the block form
\begin{equation}\label{Thm2:JT2}
    J_2(\theta;X_\bot)=
    \left[
    \begin{array}{c|ccccc}
        1 & 0 & 0 & 0 & 0 & 0 \\
        \star &   &   &   &   &   \\
        \star &   &   &   &   &   \\
        \star &   &   & D\Bigl[\Phi\!\left(\frac{\theta-\theta_0^*}{\om_0^*}\right)\Bigr]^{-1} &   &   \\
        \star &   &   &   &   &   \\
        \star &   &   &   &   &
    \end{array}
    \right].
\end{equation}
The entries marked by \(\star\) are not needed for computing the determinant.
Since \(\det D\neq 0\) and \(\det \Phi(s)\equiv 1\) for every scalar argument \(s\), it follows that \(J_2(\theta;X_\bot)\) has full rank \(6\) for every \((\theta;X_\bot)\).

By the chain rule, the Jacobian \(J_T(X)\) of the total transformation \(T(X)=T_2(T_1(X))\) is
\begin{equation}\label{Thm2:J=J1J2}
    J_T(X)=\left.J_2(\theta;X_\bot)\right|_{X_\bot=X_\bot(X)}\,J_1(X).
\end{equation}
Therefore,
\begin{equation}\label{Thm2:detJ}
    \det J_T(X)
    =
    \left.\det J_2(\theta;X_\bot)\right|_{X_\bot=X_\bot(X)}\det J_1(X)
    \equiv \det D \neq 0,
    \qquad \forall\,X.
\end{equation}
Hence, by the inverse function theorem, the transformation \(T(X)\) defined by \eqref{Thm2:X_to_Z1}--\eqref{Thm2:X_to_Z2} is locally smoothly invertible at every point of the state space of \eqref{Constrained_dynamics}, and in particular in a tubular neighborhood of the nominal motion \(q^*(t)\).

Therefore, the angle \(\theta\) together with the five scalar functions \(Z=\bigl[z_1;\dots;z_5\bigr]\) introduced by \eqref{Thm2:X_to_Z1}--\eqref{Thm2:X_to_Z2} form an alternative coordinate system for \eqref{Constrained_dynamics}.


The {second part of the proof} is devoted to rewriting the system dynamics \eqref{Constrained_dynamics} in the new coordinates \((\theta;Z)\) in order to derive \eqref{P:Thm4:theta_dynamics}--\eqref{P:Thm4:Z2_dynamics}.

To this end, let us introduce an intermediate set of variables
\(
    Y \triangleq \bigl[y_1;\dots;y_5\bigr]
\)
by
\begin{equation}\label{P:Thm2:Y}
    X_\bot(t)=\Phi_\bot\bigl(\theta(t)\bigr)Y(t).
\end{equation}
Differentiating \eqref{P:Thm2:Y} along a solution of \eqref{Constrained_dynamics}, we obtain
\begin{equation}\label{P:Thm2:dY1}
    \hbox{$\frac{d}{dt}$}X_\bot(t)
    =
    \hbox{$\frac{d}{dt}$}\Bigl[\Phi_\bot\bigl(\theta(t)\bigr)\Bigr]Y(t)
    +\Phi_\bot\bigl(\theta(t)\bigr)\hbox{$\frac{d}{dt}$}Y(t).
\end{equation}

On the other hand, by \eqref{LemAA:dXbot}--\eqref{LemAA:abf},
\begin{equation}\label{P:Thm2:dY2}
    \hbox{$\frac{d}{dt}$}X_\bot(t)
    =
    A_\bot\bigl(\theta(t)\bigr)X_\bot(t)
    +B_\bot\bigl(\theta(t)\bigr)u(t)
    +F_\bot\bigl(\theta(t),X_\bot(t)\bigr).
\end{equation}
Substituting \eqref{P:Thm2:Y} into \eqref{P:Thm2:dY2} and comparing with \eqref{P:Thm2:dY1}, we obtain
\begin{align}
    \hbox{$\frac{d}{dt}$}Y(t)
    &=
    \bigl[\Phi_\bot\bigl(\theta(t)\bigr)\bigr]^{-1}
    \Bigl\{
        A_\bot\bigl(\theta(t)\bigr)\Phi_\bot\bigl(\theta(t)\bigr)
        -\hbox{$\frac{d}{dt}$}\Phi_\bot\bigl(\theta(t)\bigr)
    \Bigr\}Y(t)
    \nonumber\\
    &\qquad
    +\bigl[\Phi_\bot\bigl(\theta(t)\bigr)\bigr]^{-1}
    \Bigl\{
        B_\bot\bigl(\theta(t)\bigr)u(t)
        +F_\bot\bigl(\theta(t),X_\bot(t)\bigr)
    \Bigr\}.
    \label{P:Thm2:dY3a}
\end{align}
Now use the chain rule:
\(
    \hbox{$\frac{d}{dt}$}\Phi_\bot\bigl(\theta(t)\bigr)
    =
    \hbox{$\frac{d}{d\theta}$}\Phi_\bot\bigl(\theta(t)\bigr)\hbox{$\frac{d}{dt}$}\theta(t).
\)
By \eqref{LemmBB:dPhi_bot(tau)},
\(
    \hbox{$\frac{d}{d\theta}$}\Phi_\bot(\theta)
    =
    \frac{1}{\om_0^*}A_\bot(\theta)\Phi_\bot(\theta),
\)
and by \eqref{LemAA:dth},
\(
    \hbox{$\frac{d}{dt}$}\theta(t)=\om_0^*+x_{5_\bot}(t).
\)
Therefore,
\begin{align*}
    \hbox{$\frac{d}{dt}$}\Phi_\bot\bigl(\theta(t)\bigr)
    &=
    \frac{1}{\om_0^*}
    A_\bot\bigl(\theta(t)\bigr)\Phi_\bot\bigl(\theta(t)\bigr)
    \bigl[\om_0^*+x_{5_\bot}(t)\bigr] \\    &
    =
    A_\bot\bigl(\theta(t)\bigr)\Phi_\bot\bigl(\theta(t)\bigr)
    +\frac{x_{5_\bot}(t)}{\om_0^*}
    A_\bot\bigl(\theta(t)\bigr)\Phi_\bot\bigl(\theta(t)\bigr).
\end{align*}
Substituting this into \eqref{P:Thm2:dY3a} gives
\begin{align}
    \hbox{$\frac{d}{dt}$}Y(t)
    &=
    \bigl[\Phi_\bot\bigl(\theta(t)\bigr)\bigr]^{-1}
    \Bigl\{
        -\frac{x_{5_\bot}(t)}{\om_0^*}
        A_\bot\bigl(\theta(t)\bigr)\Phi_\bot\bigl(\theta(t)\bigr)
    \Bigr\}Y(t)
    \nonumber\\    &\qquad
    +\bigl[\Phi_\bot\bigl(\theta(t)\bigr)\bigr]^{-1}
    \Bigl\{
        B_\bot\bigl(\theta(t)\bigr)u(t)
        +F_\bot\bigl(\theta(t),X_\bot(t)\bigr)
    \Bigr\}.
    \label{P:Thm2:dY3b}
\end{align}
Hence,
\begin{equation}\label{P:Thm2:dY3}
    \hbox{$\frac{d}{dt}$}Y(t)
    =
    R_1\bigl(\theta(t)\bigr)u(t)
    +
    R_2\bigl(\theta(t),X_\bot(t),Y(t)\bigr),
\end{equation}
\begin{equation}\label{P:Thm2:R1def}
   \hbox{where}\qquad R_1(\theta)
    \triangleq
    \bigl[\Phi_\bot(\theta)\bigr]^{-1}B_\bot(\theta)\qquad \hbox{and}
\end{equation}
\begin{equation}\label{P:Thm2:R2def}
    R_2(\theta,X_\bot,Y)
    \triangleq
    \bigl[\Phi_\bot(\theta)\bigr]^{-1}
    \Bigl\{
        F_\bot(\theta,X_\bot)
        -\frac{x_{5_\bot}}{\om_0^*}A_\bot(\theta)\Phi_\bot(\theta)Y
    \Bigr\}.
\end{equation}
The explicit expressions for \(R_1(\cdot)\) and \(R_2(\cdot)\) are obtained next.


\begin{statement}\label{P:Thm2:statement1}
The vector functions $R_1(\theta)$ and $R_2(\theta,X_\bot)$ defined in \eqref{P:Thm2:R1def} and \eqref{P:Thm2:R2def} are given by
\begin{equation}\label{P:Thm2:R12}
    R_1(\theta)=
    \left[\!\!
    \begin{array}{c}
        \frac{r_c}{J\om_0^*}\sin(\theta)\\[2mm]
        \frac{r_c}{J\om_0^*}\bigl[1-\cos(\theta)\bigr]\\[2mm]
        -\frac{r_c}{J}\\[2mm]
        0\\[2mm]
        \frac{1}{J}
    \end{array}\!\!\right],
    \qquad
    R_2(\theta,X_\bot)=
    -\frac{x_{5_\bot}}{\om_0^*}
    \left[\!\!
    \begin{array}{c}
        x_{3_\bot}\\
        x_{4_\bot}\\
        0\\
        0\\
        0
    \end{array}\!\!\right].
\end{equation}
\end{statement}

Now introduce the additional change of coordinates \((\theta;Y)\mapsto(\theta;Z)\) by
\begin{equation}\label{P:Thm2:Y_to_Z}
    Z=DY,
\end{equation}
where the constant matrix \(D\) is defined in \eqref{Lem6:change_of_coordinates}. Then the dynamics of the \(Z\)-variables follows from \eqref{P:Thm2:dY3} and \eqref{P:Thm2:R12}. Namely,
\begin{equation}\label{P:Thm2:dZ}
    \hbox{$\frac{d}{dt}$}Z
    =
    D\hbox{$\frac{d}{dt}$}Y
    =
    DR_1(\theta)u + DR_2(\theta,X_\bot).
\end{equation}

The corresponding vector functions on the right-hand side are
\begin{align}
    DR_1(\theta)
    &=
    \left[\!\!\begin{array}{ccccc}
        \frac{J\om_0^*}{r_c} & 0 & 0 & 0 & 0\\
        0 & -\frac{J\om_0^*}{r_c} & -\frac{J}{r_c} & 0 & 0\\
        0 & 0 & -\frac{J}{r_c} & 0 & 0\\
        0 & 0 & 0 & 1 & 0\\
        0 & 0 & \frac{1}{r_c} & 0 & 1
    \end{array}\!\!\right]
    \left[\!\!\begin{array}{c}
        \frac{r_c}{J\om_0^*}\sin(\theta)\\[1mm]
        \frac{r_c}{J\om_0^*}\bigl[1-\cos(\theta)\bigr]\\[1mm]
        -\frac{r_c}{J}\\[1mm]
        0\\[1mm]
        \frac{1}{J}
    \end{array}\!\!\right]
    =
    \left[\!\!\begin{array}{c}
        \sin(\theta)\\[2mm]
        \cos(\theta)\\[2mm]
        1\\[2mm]
        0\\[2mm]
        0
    \end{array}\!\!\right],
    \label{P:Thm2:DR1}
\end{align}
and
\begin{align}
    DR_2(\theta,X_\bot)
    &=
    \left[\!\!\begin{array}{ccccc}
        \frac{J\om_0^*}{r_c} & 0 & 0 & 0 & 0\\
        0 & -\frac{J\om_0^*}{r_c} & -\frac{J}{r_c} & 0 & 0\\
        0 & 0 & -\frac{J}{r_c} & 0 & 0\\
        0 & 0 & 0 & 1 & 0\\
        0 & 0 & \frac{1}{r_c} & 0 & 1
    \end{array}\!\!\right]
    \left[\!\!\begin{array}{c}
        -\dfrac{1}{\om_0^*}x_{3_\bot}x_{5_\bot}\\[2mm]
        -\dfrac{1}{\om_0^*}x_{4_\bot}x_{5_\bot}\\[2mm]
        0\\
        0\\
        0
    \end{array}\!\!\right]
    =
    \left[\!\!\begin{array}{c}
        -\dfrac{J}{r_c}x_{3_\bot}x_{5_\bot}\\[2mm]
        \phantom{-}\dfrac{J}{r_c}x_{4_\bot}x_{5_\bot}\\[2mm]
        0\\
        0\\
        0
    \end{array}\!\!\right].
    \label{P:Thm2:DR2}
\end{align}
The formula \eqref{P:Thm2:DR1} gives the vector field multiplying the control input \(u\) in \eqref{Thm3:clp_system_theta_Z123}--\eqref{Thm2:Z2_dynamics}.

It remains to rewrite the vector field \eqref{P:Thm2:DR2} in terms of the variables \((\theta;Z)\).
According to the change of coordinates \eqref{Thm2:X_to_Z2}, the last three components of \(X_\bot\) can be expressed as
\begin{align}
    \left[\!\!\begin{array}{c}
        x_{3_\bot}\\
        x_{4_\bot}\\
        x_{5_\bot}
    \end{array}\!\!\right]
    &=
    \bigl[\,0_{3\times 2}\mid I_{3\times 3}\,\bigr]X_\bot
    =
    \bigl[\,0_{3\times 2}\mid I_{3\times 3}\,\bigr]\Phi_\bot(\theta)D^{-1}Z
    \nonumber\\
    &=
    \bigl[\,0_{3\times 2}\mid \Sigma(\theta)\,\bigr]
    \left[\begin{array}{c|c}
        D_{11}^{-1} & -D_{11}^{-1}D_{12}D_{22}^{-1}\\
        \hline
        0_{3\times 2} & D_{22}^{-1}
    \end{array}\right]Z
   =
    \Sigma(\theta)D_{22}^{-1}
    \left[\!\!\begin{array}{c}
        z_3\\
        z_4\\
        z_5
    \end{array}\!\!\right]
    \nonumber\\
    &=
    \left[\!\!\begin{array}{ccc}
        \cos(\theta) & \theta\cos(\theta)-\sin(\theta) & 0\\
        \sin(\theta) & \theta\sin(\theta)+\cos(\theta) & 0\\
        0 & 0 & 1
    \end{array}\!\!\right]
    \left[\!\!\begin{array}{ccc}
        -\dfrac{r_c}{J} & 0 & 0\\
        0 & 1 & 0\\
        \dfrac{1}{J} & 0 & 1
    \end{array}\!\!\right]
    \left[\!\!\begin{array}{c}
        z_3\\
        z_4\\
        z_5
    \end{array}\!\!\right].
    \label{App:Xb123}
\end{align}

The right-hand side of \eqref{App:Xb123} can be simplified further by using the following property.

\begin{statement}\label{Lem12:z4_is_nonholonomic_constraint}
Consider the change of coordinates \eqref{Thm2:X_to_Z1}--\eqref{Thm2:X_to_Z2} for the nonlinear control system \eqref{Constrained_dynamics}. Then the fourth component \(z_4(X)\) of the newly introduced transverse coordinates \(Z\) admits the representation
\begin{equation}\label{Lem12:z4=f}
    z_4(X)\equiv \dot y\cos(\theta)-\dot x\sin(\theta),
    \qquad
    \forall\,X=\bigl[\theta;x;y;\dot\theta;\dot x;\dot y\bigr],
\end{equation}
that is, it coincides with the nonholonomic constraint \eqref{Nonholonomic_constraint}. Therefore, for every motion of the Dubins car one has \(z_4(t)\equiv 0\).
\end{statement}

Taking this property into account, \eqref{App:Xb123} becomes
\begin{equation}\label{P:Thm3:xbot_345}
    \left[\!\!\begin{array}{c}
        x_{3_\bot}\\[2mm]
        x_{4_\bot}\\[2mm]
        x_{5_\bot}
    \end{array}\!\!\right]
    =
    \left[\!\!\begin{array}{ccc}
        \cos(\theta) & \theta\cos(\theta)-\sin(\theta) & 0\\[2mm]
        \sin(\theta) & \theta\sin(\theta)+\cos(\theta) & 0\\[2mm]
        0 & 0 & 1
    \end{array}\!\!\right]
    \left[\!\!\begin{array}{c}
        -\frac{r_c}{J}z_3\\[2mm]
        0\\[2mm]
        \frac{1}{J}z_3+z_5
    \end{array}\!\!\right]
    =
    \left[\!\!\begin{array}{c}
        -\frac{r_c}{J}\cos(\theta)z_3\\[2mm]
        -\frac{r_c}{J}\sin(\theta)z_3\\[2mm]
        \frac{1}{J}z_3+z_5
    \end{array}\!\!\right].
\end{equation}
Substituting these expressions into \eqref{P:Thm2:DR2}, we obtain \eqref{Thm3:clp_system_theta_Z123}--\eqref{Thm2:Z2_dynamics}. This completes the proof.
\end{proof}


\subsection{Proof of Claim~\ref{P:Thm2:statement1}}
\begin{proof}
According to \eqref{LemAA:dXbot}--\eqref{LemAA:abf} and \eqref{LemmBB:invPhi_bot(tau)}, the vector function \(R_1(\cdot)\) is given by
\begin{align*}
    R_1
    &=
    \left[\begin{array}{c|c}
        I_{2\times 2} & -\Lambda(\theta)\Sigma^{-1}(\theta)\\
        \hline
        0_{3\times 2} & \Sigma^{-1}(\theta)
    \end{array}\right]
    \left[\begin{array}{c}
        0_{2\times 1}\\[1mm]
        -\frac{r_c}{J}\cos(\theta)\\[1mm]
        -\frac{r_c}{J}\sin(\theta)\\[1mm]
        \frac{1}{J}
    \end{array}\right] 
    =
    \left[\begin{array}{c}
        -\Lambda(\theta)\\[1mm]
        I_{3\times 3}
    \end{array}\right]
    \Sigma^{-1}(\theta)
    \left[\begin{array}{c}
        -\frac{r_c}{J}\cos(\theta)\\[1mm]
        -\frac{r_c}{J}\sin(\theta)\\[1mm]
        \frac{1}{J}
    \end{array}\right].
\end{align*}
Using the explicit expression for \(\Sigma^{-1}(\theta)\), we obtain
\begin{align*}
    \Sigma^{-1}(\theta)
    \left[\begin{array}{c}
        -\frac{r_c}{J}\cos(\theta)\\[1mm]
        -\frac{r_c}{J}\sin(\theta)\\[1mm]
        \frac{1}{J}
    \end{array}\right]
    &=
    \left[
    \begin{array}{ccc}
        \theta\sin(\theta)+\cos(\theta) & \sin(\theta)-\theta\cos(\theta) & 0\\[1.5mm]
        -\sin(\theta) & \cos(\theta) & 0\\[1.5mm]
        0 & 0 & 1
    \end{array}\right]
    \left[\begin{array}{c}
        -\frac{r_c}{J}\cos(\theta)\\[1mm]
        -\frac{r_c}{J}\sin(\theta)\\[1mm]
        \frac{1}{J}
    \end{array}\right] 
    =
    \left[\begin{array}{c}
        -\frac{r_c}{J}\\[1mm]
        0\\[1mm]
        \frac{1}{J}
    \end{array}\right].
\end{align*}
Therefore,
\begin{align*}
    R_1
    &=
    \left[\begin{array}{c}
        -\Lambda(\theta)\\[1mm]
        I_{3\times 3}
    \end{array}\right]
    \left[\begin{array}{c}
        -\frac{r_c}{J}\\[3mm]
        0\\[3mm]
        \frac{1}{J}
    \end{array}\right] \\    &
    =
    -\left[\!\!\begin{array}{ccc}
        \frac{1}{\om_0^*}\sin(\theta)
        & \frac{\theta}{\om_0^*}\sin(\theta)-\frac{2}{\om_0^*}\bigl[1-\cos(\theta)\bigr]
        & 0\\[1mm]
        \frac{1}{\om_0^*}\bigl[1-\cos(\theta)\bigr]
        & -\frac{\theta}{\om_0^*}\cos(\theta)+\frac{2}{\om_0^*}\sin(\theta)
        & 0\\[1mm]
        -1 & 0 & 0\\
        0 & -1 & 0\\
        0 & 0 & -1
    \end{array}\!\!\right]
    \left[\!\!\begin{array}{c}
        -\frac{r_c}{J}\\[3mm]
        0\\[3mm]
        \frac{1}{J}
    \end{array}\!\!\right] 
    =
    \left[\!\!\begin{array}{c}
        \frac{r_c}{J\om_0^*}\sin(\theta)\\[1mm]
        \frac{r_c}{J\om_0^*}\bigl[1-\cos(\theta)\bigr]\\[1mm]
        -\frac{r_c}{J}\\[1mm]
        0\\[1mm]
        \frac{1}{J}
    \end{array}\!\!\right].
\end{align*}
This coincides with the first formula in \eqref{P:Thm2:R12}.
Next, using the explicit form of \(F_\bot(\cdot)\), see \eqref{LemAA:dXbot}--\eqref{LemAA:abf}, the inverse matrix \(\Phi_\bot(\theta)^{-1}\), see \eqref{LemmBB:invPhi_bot(tau)}, and the definition of \(Y\), see \eqref{P:Thm2:Y}, we can rewrite \(R_2(\cdot)\) as
\begin{align*}
    R_2
    &=
    \bigl[\Phi_\bot(\theta)\bigr]^{-1}
    \left\{
        F_\bot\bigl(\theta,X_\bot\bigr)
        -\frac{x_{5_\bot}}{\om_0^*}A_\bot(\theta)\Phi_\bot(\theta)Y
    \right\} \\
    &=
    \left[\begin{array}{c|c}
        I_{2\times 2} & -\Lambda(\theta)\Sigma^{-1}(\theta)\\
        \hline
        0_{3\times 2} & \Sigma^{-1}(\theta)
    \end{array}\right]
    \left\{
        x_{5_\bot}
        \left[\!\!\begin{array}{c}
            0\\
            0\\
            -\sin\theta\bigl[x_{3_\bot}\cos\theta+x_{4_\bot}\sin\theta\bigr]\\
            \cos\theta\bigl[x_{3_\bot}\cos\theta+x_{4_\bot}\sin\theta\bigr]\\
            0
        \end{array}\!\!\right]
        -\frac{x_{5_\bot}}{\om_0^*}
        A_\bot(\theta)
X_\bot
    \right\}.
\end{align*}
Since
$A_\bot(\theta)
    \left[\!\!\begin{array}{c}
        x_{1_\bot}\\
        x_{2_\bot}\\
        x_{3_\bot}\\
        x_{4_\bot}\\
        x_{5_\bot}
    \end{array}\!\!\right]
    =
    \left[\!\!\begin{array}{c}
        x_{3_\bot}\\
        x_{4_\bot}\\
        -\om_0^*\sin\theta\cos\theta\,x_{3_\bot}-\om_0^*\sin^2\theta\,x_{4_\bot}\\
        \om_0^*\cos^2\theta\,x_{3_\bot}+\om_0^*\sin\theta\cos\theta\,x_{4_\bot}\\
        0
    \end{array}\!\!\right],$
the expression in braces simplifies to
\(x_{5_\bot}
    \left[\!\!\begin{array}{c}
        -\frac{1}{\om_0^*}x_{3_\bot}\\[1mm]
        -\frac{1}{\om_0^*}x_{4_\bot}\\[1mm]
        0\\
        0\\
        0
    \end{array}\!\!\right].
\)
Therefore,
\begin{align*}
    R_2
    &=
    x_{5_\bot}
    \left[\begin{array}{c|c}
        I_{2\times 2} & -\Lambda(\theta)\Sigma^{-1}(\theta)\\
        \hline
        0_{3\times 2} & \Sigma^{-1}(\theta)
    \end{array}\right]
    \left[\!\!\begin{array}{c}
        -\frac{1}{\om_0^*}x_{3_\bot}\\[1mm]
        -\frac{1}{\om_0^*}x_{4_\bot}\\[1mm]
        0\\
        0\\
        0
    \end{array}\!\!\right] 
    =
    -\frac{x_{5_\bot}}{\om_0^*}
    \left[\!\!\begin{array}{c}
        x_{3_\bot}\\
        x_{4_\bot}\\
        0\\
        0\\
        0
    \end{array}\!\!\right].
\end{align*}
This coincides with the second formula in \eqref{P:Thm2:R12}.
\end{proof}

\subsection{Proof of Claim~\ref{LemAA:dynamics_in_(th,Xbot)_coordinates}}
\begin{proof}
Equation \eqref{LemAA:dth} follows immediately from the definition of the fifth transverse coordinate in \eqref{Transverse_coordinates}:
\(
    x_{5_\bot}(t)=\dot\theta(t)-\om_0^*,
\)
hence
\begin{equation}\label{P:LemAA:dth}
    \dot\theta(t)=\om_0^*+x_{5_\bot}(t).
\end{equation}

Next, by \eqref{Transverse_coordinates},
\[
    \dot x(t)=x_{3_\bot}(t)+r_c\cos\bigl(\theta(t)\bigr)\dot\theta(t),
    \qquad
    x(t)=x_{1_\bot}(t)+r_c\sin\bigl(\theta(t)\bigr).
\]
Differentiating the second relation and equating the two expressions for \(\dot x(t)\), we obtain
\[
    x_{3_\bot}(t)+r_c\cos\bigl(\theta(t)\bigr)\dot\theta(t)
    =
    \hbox{$\frac{d}{dt}$}x_{1_\bot}(t)
    +r_c\cos\bigl(\theta(t)\bigr)\dot\theta(t),
\]
and therefore
\begin{equation}\label{P:LemAA:dx1bot}
    \hbox{$\frac{d}{dt}$}x_{1_\bot}(t)=x_{3_\bot}(t).
\end{equation}

Similarly,
\[
    \dot y(t)=x_{4_\bot}(t)+r_c\sin\bigl(\theta(t)\bigr)\dot\theta(t),
    \qquad
    y(t)=x_{2_\bot}(t)-r_c\cos\bigl(\theta(t)\bigr).
\]
Differentiating the second relation gives
\[
    x_{4_\bot}(t)+r_c\sin\bigl(\theta(t)\bigr)\dot\theta(t)
    =
    \hbox{$\frac{d}{dt}$}x_{2_\bot}(t)
    +r_c\sin\bigl(\theta(t)\bigr)\dot\theta(t),
\]
and hence
\begin{equation}\label{P:LemAA:dx2bot}
    \hbox{$\frac{d}{dt}$}x_{2_\bot}(t)=x_{4_\bot}(t).
\end{equation}

The fifth transverse coordinate satisfies
\begin{equation}\label{P:LemAA:dx5bot}
    \hbox{$\frac{d}{dt}$}x_{5_\bot}(t)
    =
    \ddot\theta(t)
    =
    \frac{1}{J}u(t),
\end{equation}
where the last equality follows from \eqref{Constrained_dynamics}.

We now derive the dynamics of \(x_{3_\bot}(t)\). By definition,
\(
    x_{3_\bot}(t)=\dot x(t)-r_c\cos\bigl(\theta(t)\bigr)\dot\theta(t),
\)
hence
\begin{equation}\label{P:LemAA:dx3bot_1}
    \hbox{$\frac{d}{dt}$}x_{3_\bot}(t)
    =
    \ddot x(t)
    +r_c\sin\bigl(\theta(t)\bigr)\dot\theta(t)^2
    -r_c\cos\bigl(\theta(t)\bigr)\ddot\theta(t).
\end{equation}
Using \eqref{Constrained_dynamics},
\[
    \ddot x(t)
    =
    -\Bigl[\dot y(t)\sin\bigl(\theta(t)\bigr)+\dot x(t)\cos\bigl(\theta(t)\bigr)\Bigr]
    \dot\theta(t)\sin\bigl(\theta(t)\bigr),
    \qquad
    \ddot\theta(t)=\frac{1}{J}u(t),
\]
and substituting
\[
    \dot x(t)=x_{3_\bot}(t)+r_c\cos\bigl(\theta(t)\bigr)\dot\theta(t),
    \qquad
    \dot y(t)=x_{4_\bot}(t)+r_c\sin\bigl(\theta(t)\bigr)\dot\theta(t),
\]
into \eqref{P:LemAA:dx3bot_1}, we obtain
\begin{align*}
    \hbox{$\frac{d}{dt}$}x_{3_\bot}(t)
    &=
    -\Bigl[x_{4_\bot}(t)\sin\bigl(\theta(t)\bigr)
      +x_{3_\bot}(t)\cos\bigl(\theta(t)\bigr)
      +r_c\dot\theta(t)\Bigr]
      \dot\theta(t)\sin\bigl(\theta(t)\bigr) \\
    &\qquad
    +r_c\sin\bigl(\theta(t)\bigr)\dot\theta(t)^2
    -\frac{r_c}{J}\cos\bigl(\theta(t)\bigr)u(t) \\
    &=
    -\Bigl[x_{3_\bot}(t)\cos\bigl(\theta(t)\bigr)
      +x_{4_\bot}(t)\sin\bigl(\theta(t)\bigr)\Bigr]
      \sin\bigl(\theta(t)\bigr)\dot\theta(t)
    -\frac{r_c}{J}\cos\bigl(\theta(t)\bigr)u(t).
\end{align*}
Using \eqref{P:LemAA:dth}, that is, \(\dot\theta(t)=x_{5_\bot}(t)+\om_0^*\), we arrive at
\begin{equation}\label{P:LemAA:dx3bot}
    \hbox{$\frac{d}{dt}$}x_{3_\bot}(t)
    =
    -\Bigl[x_{3_\bot}(t)\cos\bigl(\theta(t)\bigr)
      +x_{4_\bot}(t)\sin\bigl(\theta(t)\bigr)\Bigr]
      \sin\bigl(\theta(t)\bigr)\bigl[x_{5_\bot}(t)+\om_0^*\bigr]
    -\frac{r_c}{J}\cos\bigl(\theta(t)\bigr)u(t).
\end{equation}

We proceed in the same way for \(x_{4_\bot}(t)\). Since
\(
    x_{4_\bot}(t)=\dot y(t)-r_c\sin\bigl(\theta(t)\bigr)\dot\theta(t),
\)
we have
\begin{equation}\label{P:LemAA:dx4bot_1}
    \hbox{$\frac{d}{dt}$}x_{4_\bot}(t)
    =
    \ddot y(t)
    -r_c\cos\bigl(\theta(t)\bigr)\dot\theta(t)^2
    -r_c\sin\bigl(\theta(t)\bigr)\ddot\theta(t).
\end{equation}
Using \eqref{Constrained_dynamics},
\[
    \ddot y(t)
    =
    \Bigl[\dot y(t)\sin\bigl(\theta(t)\bigr)+\dot x(t)\cos\bigl(\theta(t)\bigr)\Bigr]
    \dot\theta(t)\cos\bigl(\theta(t)\bigr),
\]
and substituting the expressions for \(\dot x(t)\), \(\dot y(t)\), and \(\dot\theta(t)\), we obtain
\begin{align*}
    \hbox{$\frac{d}{dt}$}x_{4_\bot}(t)
    &=
    \Bigl[x_{4_\bot}(t)\sin\bigl(\theta(t)\bigr)
      +x_{3_\bot}(t)\cos\bigl(\theta(t)\bigr)
      +r_c\dot\theta(t)\Bigr]
      \dot\theta(t)\cos\bigl(\theta(t)\bigr) \\&\qquad
    -r_c\cos\bigl(\theta(t)\bigr)\dot\theta(t)^2
    -\frac{r_c}{J}\sin\bigl(\theta(t)\bigr)u(t) \\
    &=
    \Bigl[x_{3_\bot}(t)\cos\bigl(\theta(t)\bigr)
      +x_{4_\bot}(t)\sin\bigl(\theta(t)\bigr)\Bigr]
      \cos\bigl(\theta(t)\bigr)\dot\theta(t)
    -\frac{r_c}{J}\sin\bigl(\theta(t)\bigr)u(t).
\end{align*}
Using again \eqref{P:LemAA:dth}, we obtain
\begin{equation}\label{P:LemAA:dx4bot}
    \hbox{$\frac{d}{dt}$}x_{4_\bot}(t)
    =
    \Bigl[x_{3_\bot}(t)\cos\bigl(\theta(t)\bigr)
      +x_{4_\bot}(t)\sin\bigl(\theta(t)\bigr)\Bigr]
      \cos\bigl(\theta(t)\bigr)\bigl[x_{5_\bot}(t)+\om_0^*\bigr]
    -\frac{r_c}{J}\sin\bigl(\theta(t)\bigr)u(t).
\end{equation}

The relations \eqref{P:LemAA:dth}, \eqref{P:LemAA:dx1bot}, \eqref{P:LemAA:dx2bot}, \eqref{P:LemAA:dx3bot}, \eqref{P:LemAA:dx4bot}, and \eqref{P:LemAA:dx5bot} are exactly the equations \eqref{LemAA:dth}, \eqref{LemAA:dXbot}, and \eqref{LemAA:abf} written component-wise.
\end{proof}

\subsection{Proof of Claim~\ref{LemCC:z5=}}
\begin{proof}
We compute the constant value \(z_5(t)\equiv z_5(0)\) as a function of the initial condition \(X_0\) of a perturbed Dubins-car motion.
According to \eqref{Thm2:X_to_Z2},
\(
    Z = D\,Y,
\)
and therefore
\begin{equation}\label{P:Thm3:z5_step1}
    z_5
    =
    \bigl[0,\,0,\,0,\,0,\,1\bigr]D\,Y.
\end{equation}
Using the last row of \(D\), see \eqref{Lem6:change_of_coordinates}, we obtain
\[
    \bigl[0,\,0,\,0,\,0,\,1\bigr]D
    =
    \bigl[0,\,0,\,\tfrac{1}{r_c},\,0,\,1\bigr].
\]
Hence,
\begin{equation}\label{P:Thm3:z5_step2}
    z_5
    =
    \bigl[0,\,0,\,\tfrac{1}{r_c},\,0,\,1\bigr]Y.
\end{equation}

Next, according to \eqref{P:Thm2:Y} and \eqref{LemmBB:invPhi_bot(tau)},
\[
    Y
    =
    \bigl[\Phi_\bot(\theta)\bigr]^{-1}X_\bot
    =
    \left[\begin{array}{c|c}
        I_{2\times 2} & -\Lambda(\theta)\Sigma^{-1}(\theta)\\
        \hline
        0_{3\times 2} & \Sigma^{-1}(\theta)
    \end{array}\right]X_\bot.
\]
Therefore, \eqref{P:Thm3:z5_step2} becomes
\begin{align}
    z_5
    &=
    \left[\!\frac{1}{r_c},\,0,\,1\!\right]\Sigma^{-1}(\theta)
    \left[\!\!\begin{array}{c}
        x_{3_\bot}\\
        x_{4_\bot}\\
        x_{5_\bot}
    \end{array}\!\!\right]
    =
    \left[\!\frac{1}{r_c},\,0,\,1\!\right]
    \left[
    \begin{array}{ccc}
        \theta\sin\theta+\cos\theta & \sin\theta-\theta\cos\theta & 0\\
        -\sin\theta & \cos\theta & 0\\
        0 & 0 & 1
    \end{array}\right]
    \left[\!\!\begin{array}{c}
        x_{3_\bot}\\
        x_{4_\bot}\\
        x_{5_\bot}
    \end{array}\!\!\right] \nonumber\\
    &=
    \frac{\theta\sin\theta+\cos\theta}{r_c}\,x_{3_\bot}
    +\frac{\sin\theta-\theta\cos\theta}{r_c}\,x_{4_\bot}
    +x_{5_\bot}.
    \label{P:Thm3:z5_step3}
\end{align}

Now substitute the definitions of the transverse coordinates from \eqref{Transverse_coordinates}:
\[
    x_{3_\bot}=\dot x-r_c\cos\theta\,\dot\theta, \qquad
    x_{4_\bot}=\dot y-r_c\sin\theta\,\dot\theta, \qquad
    x_{5_\bot}=\dot\theta-\om_0^*.
\]
Then \eqref{P:Thm3:z5_step3} becomes
\begin{align*}
    z_5
    &=
    \frac{\theta\sin\theta+\cos\theta}{r_c}\bigl(\dot x-r_c\cos\theta\,\dot\theta\bigr)
    +\frac{\sin\theta-\theta\cos\theta}{r_c}\bigl(\dot y-r_c\sin\theta\,\dot\theta\bigr)
    +\dot\theta-\om_0^* \\
    &=
    \frac{\theta}{r_c}\bigl(\dot x\sin\theta-\dot y\cos\theta\bigr)
    +\frac{1}{r_c}\bigl(\dot x\cos\theta+\dot y\sin\theta\bigr) \\
    &\qquad
    -\bigl[(\theta\sin\theta+\cos\theta)\cos\theta+(\sin\theta-\theta\cos\theta)\sin\theta\bigr]\dot\theta
    +\dot\theta-\om_0^* \\
    &=
    \frac{\theta}{r_c}\bigl(\dot x\sin\theta-\dot y\cos\theta\bigr)
    +\frac{1}{r_c}\bigl(\dot x\cos\theta+\dot y\sin\theta\bigr)
    -\dot\theta+\dot\theta-\om_0^* \\
    &=
    \frac{\theta}{r_c}\bigl(\dot x\sin\theta-\dot y\cos\theta\bigr)
    +\frac{1}{r_c}\bigl(\dot x\cos\theta+\dot y\sin\theta\bigr)
    -\om_0^*.
\end{align*}
By the constraint \eqref{Nonholonomic_constraint},
$\dot x\sin\theta-\dot y\cos\theta=0,$
and by \eqref{kinematics_for_XY}--\eqref{kinematics_constraint_for_V},
$\dot x\cos\theta+\dot y\sin\theta=v_0.$
Hence,
\begin{equation}\label{P:Thm3:z5=}
    z_5=\frac{v_0}{r_c}-\om_0^*.
\end{equation}
Finally, since \(v_0=\sqrt{\dot x_0^2+\dot y_0^2}\) for the considered Dubins-car motion, \eqref{P:Thm3:z5=} coincides with \eqref{LemCC:z5=C}.
\end{proof}

\subsection{Proof of Claim~\ref{LemBB:on_dPhi_nonlinear}} 
\begin{proof}
One readily verifies that
\begin{equation}\label{P:LemBB:dPhi}
    \hbox{$\frac{d}{d\tau}$}\Phi_\bot(\tau)=
    \left[\begin{array}{ccccc}
        0 & 0 & \frac{1}{\om_0^*}\cos\tau
          & \frac{\tau}{\om_0^*}\cos\tau-\frac{1}{\om_0^*}\sin\tau & 0\\
        0 & 0 & \frac{1}{\om_0^*}\sin\tau
          & \frac{\tau}{\om_0^*}\sin\tau+\frac{1}{\om_0^*}\cos\tau & 0\\
        0 & 0 & -\sin\tau
          & -\tau\sin\tau & 0\\
        0 & 0 & \cos\tau
          & \tau\cos\tau & 0\\
        0 & 0 & 0 & 0 & 0
    \end{array}\right].
\end{equation}
On the other hand,
\[\begin{array}{l}
    A_\bot(\tau)\Phi_\bot(\tau)\\[2mm]
    =
    \left[\!\!
    \begin{array}{ccccc}
        0 & 0 & 1 & 0 & 0\\[1mm]
        0 & 0 & 0 & 1 & 0\\[1mm]
        0 & 0 & -\om_0^*\sin\tau\cos\tau & -\om_0^*\sin^2\tau & 0\\[1mm]
        0 & 0 & \om_0^*\cos^2\tau & \om_0^*\sin\tau\cos\tau & 0\\[1mm]
        0 & 0 & 0 & 0 & 0
    \end{array}
    \!\!\right]
    \left[\!\!
    \begin{array}{ccccc}
        1 & 0 & \frac{1}{\om_0^*}\sin\tau
          & \frac{\tau}{\om_0^*}\sin\tau+\frac{2\bigl(\cos\tau-1\bigr)}{\om_0^*} & 0\\
        0 & 1 & \frac{\bigl(1-\cos\tau\bigr)}{\om_0^*}
          & -\frac{\tau}{\om_0^*}\cos\tau+\frac{2}{\om_0^*}\sin\tau & 0\\
        0 & 0 & \cos\tau
          & \tau\cos\tau-\sin\tau & 0\\
        0 & 0 & \sin\tau
          & \tau\sin\tau+\cos\tau & 0\\
        0 & 0 & 0 & 0 & 1
    \end{array}
    \!\!\right] \\[2mm]
    =
    \left[\begin{array}{ccccc}
        0 & 0 & \cos\tau
          & \tau\cos\tau-\sin\tau & 0\\
        0 & 0 & \sin\tau
          & \tau\sin\tau+\cos\tau & 0\\
        0 & 0 & -\om_0^*\sin\tau
          & -\om_0^*\tau\sin\tau & 0\\
        0 & 0 & \om_0^*\cos\tau
          & \om_0^*\tau\cos\tau & 0\\
        0 & 0 & 0 & 0 & 0
    \end{array}\right]=\om_0 \hbox{$\frac{d}{d\tau}$}\Phi_\bot(\tau).
\end{array}\]
Hence the claimed differential identity holds.

The remaining relations in the statement follow from direct computation of the inverse matrix and standard block-matrix inversion formulas.
\end{proof}

\subsection{Proof of Claim~\ref{Lem12:z4_is_nonholonomic_constraint}} 
\begin{proof}
According to \eqref{P:Thm2:Y_to_Z} and \eqref{P:Thm2:Y},
\begin{align*}
    z_4
    &=
    \bigl[0,\,0,\,0,\,1,\,0\bigr]\,D\,Y 
    =
    \bigl[0,\,0,\,0,\,1,\,0\bigr]
    \left[\!\!\begin{array}{ccccc}
        \frac{J\om_0^*}{r_c} & 0 & 0 & 0 & 0\\
        0 & -\frac{J\om_0^*}{r_c} & -\frac{J}{r_c} & 0 & 0\\
        0 & 0 & -\frac{J}{r_c} & 0 & 0\\
        0 & 0 & 0 & 1 & 0\\
        0 & 0 & \frac{1}{r_c} & 0 & 1
    \end{array}\!\!\right]
    \left[\!\!\begin{array}{c}
        y_1\\[2mm]
        y_2\\[2mm]
        y_3\\[2mm]
        y_4\\[2mm]
        y_5
    \end{array}\!\!\right] 
    = y_4.
\end{align*}
On the other hand,
\begin{align*}
    y_4
    &=
    \bigl[0,\,0,\,0,\,1,\,0\bigr]\bigl[\Phi_\bot(\theta)\bigr]^{-1}X_\bot 
    =
    \bigl[0,\,0,\,0,\,1,\,0\bigr]
    \left[\begin{array}{c|c}
        I_{2\times 2} & -\Lambda(\theta)\Sigma^{-1}(\theta)\\
        \hline
        0_{3\times 2} & \Sigma^{-1}(\theta)
    \end{array}\right]X_\bot \\
    &=
    \bigl[0,\,1,\,0\bigr]\Sigma^{-1}(\theta)
    \left[\!\!\begin{array}{c}
        x_{3_\bot}\\
        x_{4_\bot}\\
        x_{5_\bot}
    \end{array}\!\!\right] 
    =
    \bigl[0,\,1,\,0\bigr]
    \left[\!\!\begin{array}{ccc}
        \theta\sin\theta+\cos\theta & \sin\theta-\theta\cos\theta & 0\\
        -\sin\theta & \cos\theta & 0\\
        0 & 0 & 1
    \end{array}\!\!\right]
    \left[\!\!\begin{array}{c}
        x_{3_\bot}\\
        x_{4_\bot}\\
        x_{5_\bot}
    \end{array}\!\!\right] \\
    &=
    x_{4_\bot}\cos\theta-x_{3_\bot}\sin\theta.
\end{align*}
Therefore,
\(
    z_4=x_{4_\bot}\cos\theta-x_{3_\bot}\sin\theta.
\)
Using the definitions of $x_{3_\bot}$ and $x_{4_\bot}$ from \eqref{Transverse_coordinates}, we obtain
\begin{align*}
    z_4
    &=
    \bigl[\dot y-r_c\sin\theta\,\dot\theta\bigr]\cos\theta
    -\bigl[\dot x-r_c\cos\theta\,\dot\theta\bigr]\sin\theta 
    =
    \dot y\cos\theta-\dot x\sin\theta.
\end{align*}
This proves \eqref{Lem12:z4=f}.
\end{proof}
\end{document}